\documentclass[11pt,a4paper]{article}

\usepackage{fullpage}
\usepackage{amsmath,amsfonts,amssymb}
\usepackage[dvipsnames]{xcolor}
\usepackage{epsfig}
\usepackage{hyperref}
\usepackage{cite}

\newtheorem{theorem}{Theorem}[section]
\newtheorem{claim}[theorem]{Claim}
\newtheorem{corollary}[theorem]{Corollary}
\newtheorem{definition}[theorem]{Definition}

\newtheorem{lemma}[theorem]{Lemma}
\newtheorem{property}[theorem]{Property}

\newcommand{\qed}{\hfill $\Box$ \bigbreak}
\newenvironment{proof}{\noindent{\bf Proof.~}}{\qed}

\DeclareMathOperator{\Div}{Div}
\DeclareMathOperator{\name}{names}
\DeclareMathOperator{\val}{values}
\DeclareMathOperator{\skel}{skel}

\newcommand{\cA}{\ensuremath{\mathcal{A}}}
\newcommand{\cB}{\ensuremath{\mathcal{B}}}

\newcommand{\cI}{\ensuremath{\mathcal{I}}}
\newcommand{\cK}{\ensuremath{\mathcal{K}}}
\newcommand{\cL}{\ensuremath{\mathcal{L}}}
\newcommand{\cO}{\ensuremath{\mathcal{O}}}
\newcommand{\cP}{\ensuremath{\mathcal{P}}}

\DeclareFontShape{OT1}{cmr}{bx}{sc}{<-> cmbcsc10}{} 
\newcommand{\kagraph}{\ensuremath{\mathsf{G}}}
\newcommand{\kagraphr}{\ensuremath{\mathsf{G}_{\leq r}}}

\newcommand{\wcirc}{\textcolor{black}{\ensuremath\circ}}
\newcommand{\gcirc}{\textcolor{gray}{\ensuremath\bullet}}
\newcommand{\bcirc}{\textcolor{black}{\ensuremath\bullet}}
\newcommand{\local}{\textnormal{\textsc{local}}}
\newcommand{\cgst}{\textnormal{\textsc{congest}}}
\newcommand{\knowall}{\textnormal{\textsc{know-all}}}

\begin{document}

\title{A Topological Perspective on Distributed Network Algorithms%
\thanks{Some of the results in this paper were presented in an invited talk in SIROCCO 2019 conference~\cite{CastanedaFPRRT19talk}.\newline}}

\author{
	Armando Casta\~neda\thanks{Supported by UNAM-PAPIIT IN108720.} \\ {\small UNAM, Mexico}
	\and Pierre Fraigniaud\thanks{Supported by ANR projects DESCARTES and FREDA. Additional support from INRIA project GANG.} \\ {\small CNRS and Univ.\ de Paris}
	\and Ami Paz\thanks{Supported by the Austrian Science Fund (FWF): P 33775-N, Fast Algorithms for a Reactive Network Layer.}
	\\ {\small CS Faculty, Univ. of Vienna}
	\and Sergio Rajsbaum\thanks{Supported by project UNAM-PAPIIT  IN109917.} \\ {\small \;\;\;\;\;\;\; UNAM, Mexico \;\;\;\;\;\;\;}
	\and  Matthieu Roy\\ {\small CNRS, France}
	\and Corentin Travers\thanks{Supported by ANR projects DESCARTES and FREDA.} \\ {\small Univ. of Bordeaux and CNRS}
}
\date{}

\maketitle

\begin{abstract}
More than two decades ago, combinatorial topology was shown to be useful for analyzing distributed fault-tolerant algorithms in shared memory systems and in message passing systems. In this work, we show that combinatorial topology can also be useful for analyzing distributed algorithms in failure-free networks of arbitrary structure. To illustrate this, we analyze consensus, set-agreement, and approximate agreement in networks, and derive lower bounds for these problems under classical computational settings, such as the \local{} model and dynamic networks.

~

\end{abstract}

\thispagestyle{empty}
\setcounter{page}{1}
\newpage

\section{Introduction}

\subsection{Context and Objective}

A breakthrough in distributed computing was obtained in the 1990's, when \emph{combinatorial topology}, a branch of Mathematics extending graph theory to higher dimensional objects, was shown to provide a framework in which a large variety of distributed computing models can be studied~\cite{BorowskyG93,HerlihyS93,SaksZ93}. Combinatorial topology provides a powerful arsenal of tools, which considerably expanded our understanding of the solvability and complexity of many distributed problems~\cite{AttiyaCHP19,CastanedaR10,CastanedaR12,HerlihyS99}. We refer to the book by Herlihy et al.~\cite{HerlihyKR13} for an extended and detailed description of combinatorial topology applied to distributed computing, in a wide variety of settings. 

In a nutshell, combinatorial topology allows us to represent all possible executions of a distributed algorithm, along with the relations between them, as a single mathematical object, whose properties reflect the solvability of a problem. Combinatorial topology was primarily used to study failure-tolerant distributed computing in the context of  shared memory and message passing systems, but not in the context of failure-free systems in which the presence of a \emph{network} connecting the processing elements needs to be taken into account. However, a large portion of the study of distributed computing requires to take into account the structure of the network connecting the processors, e.g, when studying \emph{locality}~\cite{Peleg2000}. This paper is a first attempt to approach distributed network computing through the lens of combinatorial topology. 

The base of the topological approach for distributed computing consists of modeling all possible input configurations (resp., output configurations) as a single object called input \emph{complex} (resp., output complex), and specifying a task as a relation between the input and output complexes. Moreover, computation in a given model results in a topological deformation that modifies the input complex into another complex called the \emph{protocol complex}. The fundamental result of combinatorial topology applied to distributed computing~\cite{HerlihyKR13} is that a task is solvable in a computational model if and only if there exists a \emph{simplicial map}, called a \emph{decision map},  from the protocol complex to the output complex, that agrees with the specification of the task. In other words, for every input configuration, (1)~the execution of the algorithm leads the system into one or many configurations, forming a subcomplex of the protocol complex, and (2)~the decision map should map every configuration in this subcomplex (i.e., each of its \emph{simplexes}) into a configuration of the output complex, that is legal for the given input configuration, with respect to the specification of the task. 

Understanding the power and limitation of a distributed computing model with respect to solving a given task requires to understand under which condition the decision map exists. This requires to understand the nature of topological deformations of the input complex resulting from the execution of an algorithm, and the outcome of this deformation, i.e., the protocol complex. That is, one needs to establish the connections between the distributed computing model at hand, and the topological deformations incurred by the input complex in the course of a computation under this model.  

The connections between the computational models and the topological deformations are now well understood for several distributed computing models. For instance, in shared-memory wait-free systems, the protocol complex results from the input complex by a series of specific \emph{subdivisions} of its simplexes. 
The impossibility result for consensus in shared-memory wait-free systems is a direct consequence of this fact, as the input complex of consensus is connected, subdivisions maintain connectivity, but the output complex of consensus is not connected --- this prevents the existence of a decision map, independently of how long the computation proceeds. Similarly, in shared-memory $t$-resilient systems, the protocol complex results from the input complex not only by a series of specific subdivisions, but also by the appearance of some \emph{holes} in the course of the computation. 
These holes enable the existence of a decision map in the case of $(t+1)$-set-agreement, but are not sufficient to enable the existence of a decision map for consensus, as long as $t\geq 1$. And indeed, the FLP result~\cite{FischerLP85} implies that consensus is not solvable in asynchronous systems even in the presence of at most one failure.

This paper addresses the following issues: What is the nature of the topological deformations incurred by the input complex in the context  of network computing, i.e., when nodes are bounded to interact only with nearby nodes according to some graph metric? And, what is the impact of these deformations on the ability to solve tasks efficiently (e.g., locally) in networks? As a first step towards answering these questions in general, we tackle them in the framework of synchronous failure-free computing, which is actually the framework in which most studies of distributed network computing are conducted~\cite{Peleg2000}. 

\subsection{Our Results}

The main contribution of this paper is in studying the topological model of distributed computing in networks. In particular, we show that the protocol complex involves deformations that were not observed before in the context of distributed computing, deformations which we call \emph{scissor cuts}. 
These cuts depend on the structure of the underlying network governing the way the information flows between nodes.
We show that understanding the scissor cuts is useful for deriving lower bounds on agreement tasks. For this purpose, we model the way information flows between nodes in the network by the so-called \emph{information-flow graph}, and establish tight connections between structural properties of this graph, and the ability to solve agreement tasks in the network. For instance, we show that if the domination number of the information-flow graph is at least $k+1$, then the protocol complex is at least $(k-1)$-connected, and if the protocol complex is at least $(k-1)$-connected, then $k$-set agreement is not solvable. 

Our results apply to a rather general model: we assume a dynamic graph, i.e., a sequence $(G_t)_{t\geq 1}$ of graphs on the same set of nodes, where the graph $G_t$ represents the communication links that can be utilized in round $t$.
To analyze the state of the system after $r$ rounds, we study the question of which nodes may have heard from which other nodes: a message from a node $u$ could have arrived to a node $v$ if there is a \emph{temporal path} from $u$ to $v$, i.e., a sequence $u_{i_0},\dots,u_{i_\ell}$ of nodes with $1\leq i_0<i_1<\dots < i_\ell \leq r$ such that $u_{i_0}=u$, $u_{i_\ell}=v$, and $(u_{i_j},u_{i_{j+1}})\in E(G_{i_j})$ for every $j=0,\dots,\ell-1$.
The information flow graph represents these relations, for a given~$r$, by having a node for each node of the system, and an edge between two nodes that have a temporal path between them.

In this model, we take rather strong assumptions on the knowledge of the processes. In the \knowall{} model that we use, processing nodes are assumed to have \emph{structure awareness}, that is, every process is given the (possibly evolving) structure of the network, the position it occupies in this network, and the positions occupied by all the other processes in the network.
It follows that our results imply lower bounds for solving agreement problems in weaker models, such as the classical \local{} model, as well as dynamic networks.

For instance, a consequence of our results is that, in the \local{} model, solving $k$-set agreement in a network requires at least $r$~rounds, where $r$ is the smallest integer such that the $r$-th power of the network has domination number at most~$k$ (two nodes are adjacent in the $r$-th power of a graph if their distance in the graph is at most~$r$). Similarly, we show that solving $k$-set agreement against a dynamic network adversary~$A$ requires at least $r$~rounds, where $r$ is the smallest integer such that, for every sequence of graphs $(G_t)_{1\leq t\leq r}$ produced by $A$, this sequence has ``temporal'' dominating number at most $k$,
i.e., there is a set $S$ of at most $k$ nodes such that each node $v$ in the graph has a node $s\in S$ with a temporal path from $s$ to $v$ in~$(G_t)_{1\leq t\leq r}$.

Applying the topological approach to network computing also yields fine grained results. 
We show that for every instance of the \knowall{} model with $n$ nodes, and for every $\epsilon<\frac1{n-1}$,
if consensus is not solvable in $r$ rounds, then $\epsilon$-approximate agreement is also not solvable in $r$ rounds. This bound is tight, in the sense that there exist an instance of the \knowall{} model and a value $r$ where consensus is not solvable in $r$ rounds, while $\frac1{n-1}$-approximate agreement is solvable in $r$ rounds.

\subsection{Related Work}

The deep connections between combinatorial topology and distributed computing were concurrently and independently identified in~\cite{BorowskyG93,HerlihyS93} and~\cite{SaksZ93}. Since then, numerous outstanding results were obtained using combinatorial topology for various types of tasks, including agreement tasks such as consensus and set-agreement~\cite{SakavalasTseng2019}, and symmetry breaking tasks such as renaming~\cite{AttiyaCHP19,CastanedaR10,CastanedaR12}. A recent work~\cite{AlistarhAEGZ18} provides evidence that topological arguments are sometimes necessary. All these results are obtained in the asynchronous shared memory model with crash failures, but combinatorial topology can also be applied to Byzantine failures~\cite{MendesTH14}. Works on message passing models consider only complete communication graphs~\cite{ChaudhuriHLT00,HerlihyRT09}. Recent results show that combinatorial topology can also be applied in the analysis of mobile computing~\cite{ACPR19}, demonstrating the generality and flexibility of the topological framework applied to distributed computing. The book~\cite{HerlihyKR13} provides an extensive introduction to combinatorial topology applied to distributed computing. 

In contrast, distributed network computing has not been impacted by combinatorial topology. This domain of distributed computing is extremely active and productive this last decade, analyzing a large variety of network problems in the \local{} model~\cite{Peleg2000}, capturing the ability to solve tasks locally in networks. (The \cgst{} model has also been subject of tremendous progresses, but this model does not support full-information protocols, and thus is out of the scope of our paper). We refer to~\cite{BarenboimEG18,BarenboimEPS12,BrandtFHKLRSU16,ChangLP18,FischerGK17,Ghaffari16,GhaffariKM17,HarrisSS16,Suomela13} for a non exhaustive list of achievements in this context. However, all these achievements were based on an operational approach, using sophisticated algorithmic techniques and tools solely from graph theory. Similarly, the existing lower bounds on the round-complexity of tasks in the \local{} model~\cite{KuhnMW16,Linial92,BrandtFHKLRSU16,GoosHS17,BalliuBHORS19} were obtained using graph theoretical arguments only. The question of whether adopting a higher dimensional approach by using topology would help in the context of local computing, be it for a better conceptual understanding of the algorithms, or providing stronger technical tools for proving lower bounds, is, to our knowledge, entirely open.

Similarly to (static) distributed network computing, the fundamental research on dynamic networks~\cite{CasteigtsFGSY15,CasteigtsFQS12,KuhnO2011,BhadraF12} has rarely been impacted by combinatorial topology. Relevant works in this framework study consensus~\cite{CouloumaGP15,KuhnOM11}, set-agreement~\cite{BielyRSSW18,GodardP16} and approximate agreement~\cite{Charron-BostFN15,FuggerNS18,Charron-BostFN16}. 
Recently, techniques of \emph{point set topology}~\cite{AlpernS85} were used to characterize the solvability of consensus in the message adversary model~\cite{Nowak0W19}.
We also refer to~\cite{Charron-BostS09,KuhnLO10,RajsbaumRT08} which analyze distributed computation in a model where all processes broadcast messages at each round, but the recipients of these messages are defined by a graph which may change from round to round. The information-flow graph introduced and analyzed in our paper can be viewed as an abstraction of computation in dynamic networks, since the information-flow graph contains a summary of how information was transmitted among processes in the network during some interval of time.

\section{The \knowall{} Model}
\label{sec:model}

In this section, we describe the \knowall\/  model of computation, and show that this model is stronger than classical distributed computing models, including the \local\/ model and dynamic graphs. That is, any lower bound for the \knowall\/  model implies lowers bounds for these classical models. 

\subsection{Definition of the Model} 

We consider a set of $n$ synchronous fault-free processes, with distinct names from the set $[n]=\{1,\dots,n\}$.  

\begin{definition} 
An instance of the \knowall{} model is a sequence $\kagraph=(G_t)_{t\geq 1}$ of $n$-node directed graphs, with nodes labeld by the integers in~$[n]$.
We identify this sequence with a computational model of $n$ processes, where each process is a priori given its name $p\in[n]$ and the sequence $\kagraph$ of graphs.
The computation in this model proceeds in synchronous rounds, where at round $t$ process $p$ can transmit information to process $q$ only if the graph $G_t$ contains the arc~$(p,q)$.
\end{definition} 
In the \knowall\/ model, every process is thus given the complete knowledge of the communication patterns occurring at any time: the communication graph, the names of the processes, and their locations.
The only uncertainty is regarding the \emph{inputs} to the processes.
Hence, this model is particularly suited for studying input-output tasks, as defined next, such as consensus and list coloring. 
On the other hand, other classical tasks such as $(\Delta+1)$-coloring are trivial in this model.

\paragraph{Input-Output Tasks.}

An input-output  task $(I,O,F)$ in the $n$-process \knowall\/ model is described by a set $I$ of input values, a set $O$ of output values, and a mapping
\[F: I^n\to 2^{O^n}\]
specifying, for every $n$-tuple of input values, the set of possible legal $n$-tuples of output values. 
The input of a process $p\in[n]$ consists solely of a value~$in(p)\in I$, and its output is a value $out(p)\in O$.
While all the problems considered in this work allow any combination of input values, there are other problems in the literature which do not allow this, and assume, e.g., \emph{unique} input values, or an initial proper node coloring. In these cases, 
an algorithm solving the task is not obliged to provide any guarantees,
and accordingly, we will map any such tuple of input values to all tuples of output values.

Fix $n\geq 2$, an input-output task $(I,O,F)$, and an instance $\kagraph=(G_t)_{t\geq 1}$ of the \knowall\/ model.
Given an $n$-tuple of input values in $I$ with each process $p\in[n]$ having input~$in(p)\in I$,
an $r$-round protocol consists of having the processes synchronously communicate over the graphs $G_1,\dots,G_r$, and then having each process decide on an output value in $O$ based on the information received in these $r$ rounds of communication.
Such a protocol solves the task $(I,O,F)$ if for every $n$-tuple of input values in $I$, the values produced by the processes form an $n$-tuple of output values in $O$ which is legal by $F$.
If such a protocol exists, we say that the task $(I,O,F)$ is solvable in $r$ rounds in~$\kagraph$.

\paragraph{Flooding protocol.} A distributed algorithm solving a task has two components: a \emph{communication protocol} $\Pi$ enabling each process to gather information about the inputs of other processes, and a \emph{decision function}~$f$ that maps the gathered information to an output value. Without loss of generality, in the \knowall{} model, we can consider solely \emph{flooding} communication protocols: at round $t$, every process~$p\in [n]$ sends all the name-input pairs it is aware of to all its out-neighbors in~$G_t$, that is, it sends the pair $(p,in(p))$, and all the pairs it has received in the previous rounds. After a certain number of communication rounds, the process takes a decision based on the set of pairs it is aware of.

In general, a distributed algorithm in our setting could be much more sophisticated, and have each process send more information, or take a decision, e.g., based on the time when each message arrived or on knowledge regarding messages received by other processes.
The most general form of such algorithms is \emph{full-information} protocols, in which at every round $t=1,\dots,r$, every process sends its entire history to all its neighbors. 
While this stands in contrast with flooding protocols defined above, where each process only sends the name-input pairs it is aware of, the following lemma shows that considering only flooding protocols does not restrict the power of the \knowall\/ model. 
At the hart of its proof lies the fact that in the \knowall\/ model, each process knows the entire topology of the network in each round, including the names of the processes and their locations, and thus the decision function $f$ can depend on this information.

\begin{lemma}\label{lem:flooding-is-universal}
Every input-output  task solvable in $r$ rounds in the \knowall{} model is also solvable in $r$ rounds using a flooding protocol.
\end{lemma}

\begin{proof}
Let $(I,O,F)$ be an input-output task solvable in $r$ rounds in the \knowall{} model $\kagraph=(G_t)_{t\geq 1}$ using some algorithm $(\Pi,f)$. Then $(I,O,F)$ is also solvable in $r$ rounds using the full-information protocol,
so it is enough to consider the case where $\Pi$ is the full-information protocol. 

Let us execute the flooding protocol for $r$ rounds. Since every process is aware of the sequence of directed graphs $G_t$, $1 \leq t \leq r$, it can reconstruct the sequence of messages it would have received in the full-information protocol~$\Pi$, based on the messages received during the execution of the flooding protocol. It follows that every process is able to apply the decision function $f$ of the algorithm $(\Pi,f)$ after $r$ rounds.  
\end{proof}

Thanks to Lemma~\ref{lem:flooding-is-universal}, in order to analyze algorithms in the  \knowall\/ model, we can focus on the \emph{information-flow graph}, which describes the execution of a flooding protocol in the \knowall\/ model, defined as follows. 

\begin{definition}
Let $\kagraph=(G_t)_{t \geq 1}$ be an instance of the \knowall{} model, and let $r\geq 0$. The \emph{information-flow graph}  associated to $\kagraph$ after $r$ rounds is the directed graph $\kagraphr$ whose $n$ nodes are labeled by integers in $[n]$, and there is an arc $(p,q)$ from $p$ to $q$ in $\kagraphr$ if and only if $q$ receives the pair $(p,in(p))$ when flooding in~$\kagraph$ for $r$ rounds, i.e., there is a temporal path from $p$ to $q$ of length at most $r$.
\end{definition}

\subsection{Relation to Other Models}
\label{subsec:other-models}

Recall that, in the \local{} model~\cite{Peleg2000}, the $n$ processes  synchronously communicate using a fixed (possibly directed) communication graph~$G$, which is unknown to the processes. Every process is given a distinct identifier taken from some range of integers. 
As the execution goes by, the processes learn the topology of $G$. More precisely, after $r$ rounds of communication,
each process is aware of its $r$-neighborhood in $G$, including the identifiers of the processes in this neighborhood, and their possible inputs.
A protocol solving a given task in the \local{} model is an \emph{$r$-round} protocol if every process makes a decision in at most $r$ rounds, for every assignment of inputs and identifiers  to the processes. 

\begin{property}
\label{lemma-from-local-to-knowall}
If an input-output task is solvable in $r$ rounds in the \local{} model with communication graph $G$, then it is solvable in $r$ rounds  in the \knowall{} model $\kagraph=(G_t)_{t\geq 1}$ with $G_t=G$ for $t\geq 1$. 
\end{property}

In the framework of dynamic networks, an \emph{adversary} $A$ is a collection of infinite sequences~$(G_t)_{t\geq 1}$ of directed graphs on the same set $[n]$ of vertices, representing processes. When the execution of a protocol starts, the adversary picks any of the sequences $(G_t)_{t\geq 1}\in A$, and the processes synchronously communicate following the graphs in the sequence, that is, at round $t\geq 1$ processes exchange messages along the edges of $G_t$. A  protocol solves a task in $r$ rounds against a dynamic network adversary $A$ if, for every $(G_t)_{t\geq 1}\in A$, the protocol solves the task in $(G_t)_{t\geq 1}$. 

\begin{property}
\label{lemma-from-dynamic-to-knowall}
If an input-output task is solvable in $r$ rounds against a dynamic network adversary~$A$, then, for every $\kagraph=(G_t)_{t\geq 1}\in A$, the task is solvable in $r$ rounds  in the \knowall{} model $\kagraph$. 
\end{property}

\section{The Topology of Computing in the \knowall{} Model}

In this section, we describe the protocol complex related to distributed computing in the \knowall{} model, and establish a necessary and sufficient condition for input-output task solvability in this model. 
Most of the topological notions we use are routinely applied for studying distributed systems, as described in the book~\cite{HerlihyKR13}. 
The reader unfamiliar with these notions may read the brief recap in Appendix~\ref{appendix:topological basics}, which contains the topological notions required for the statement of our results. 

One exception to this is a new type of mappings, which we define next.
Recall that, given a simplex $\sigma$ of a complex $\cA$, the \emph{star of $\sigma$ in $\cA$}, denoted $\mathrm{St}(\sigma,\cA)$, is the sub-complex of $\cA$ induced by all its facets containing~$\sigma$. 

We say that a mapping $f:\cA\to 2^\cB$ is a \emph{spreading map} if, for every simplex $\sigma\in\cA$, 
\[
f(\sigma)\subseteq \bigcup_{\phi \, \in \, \mathrm{St}(\sigma,\cA)} f(\phi). 
\]
Note that every carrier map $f:\cA\to 2^\cB$ is a spreading map, because carrier maps enforce that  $f(\sigma)\subseteq f(\phi)$ whenever $\sigma\subseteq\phi$. However, a spreading map  does not need to be a carrier map, as the image of a simplex $\sigma = \phi'\cap\phi''$ may have its image spread partially in $f(\phi')$, and partially in $f(\phi'')$, with $f(\sigma)\not\subseteq f(\phi')$. 

While computation in asynchronous, crash-prone systems is modeled using carrier maps (specifically, subdivisions), the computation in the \knowall{} model is modeled by a different type of spreading maps, which we call scissor cuts. These are defined below.

\subsection{The Protocol Complex of the \knowall{} Model}
\label{subsec:topologicalcomputing}

Given a  distributed computing task $(I,O,F)$ to be solved in the \knowall{} model, two complexes play a major role in this framework, the \emph{input complex}, denoted by $\cI$, and the \emph{output complex}, denoted by~$\cO$.
The input complex $\cI$ is the \emph{pseudosphere} $\Psi(\{1,\ldots,n\},I)$, with the set of facets
\[
\big \{ \{(1,v_1),\dots,(n,v_n)\} : \forall i\in[n], v_i \in I \big \}. 
\]
The set of all facets of the output complex $\cO$ is
\[
\big \{ \{(1,v'_1),\dots,(n,v'_n)\} : \forall i\in[n], v'_i \in O, \;\mbox{and} \;  \exists v \in I^n, (v'_1,\dots,v'_n) \in F(v) \big \}. 
\]
Note that the output complex includes only combinations of output values that are legal with respect to the problem at hand. Note also that the input and output complexes do not depend on the communication medium considered, and that both complexes are pure---all their facets have the same dimension.

For instance, in the case of binary consensus in an $n$-process system (see Figure~\ref{fig:scisor}),  the set of facets of the input complex is 
\[
\big \{ \{(1,v_1),\dots,(n,v_n)\}:  \forall i\in[n], v_i \in\{0,1\}\big \}.
\]
This complex is topologically equivalent to the $(n-1)$-dimensional sphere $S^{n-1}$. 
(Formally, the topological spaces are \emph{homeomorphic}, that is, there is a continuous bijection from $|\cI|$ to $S^{n-1}$, whose inverse is continuous as well.)
The output complex of binary consensus is composed of two disjoints $(n-1)$-facets, $\tau_0$ and~$\tau_1$, defined by
\[
\tau_0 = \{(1,0),\dots,(n,0)\}, \; \mbox{and} \; \tau_1=\{(1,1),\dots,(n,1)\}.
\]

One can rephrase the operational definition $(I,O,F)$ of input-output task in the framework of combinatorial topology as follows: a task $(\cI,\cO,\Delta)$  is simply described by a carrier map 
\[
\Delta:\cI\to\cO,
\] 
i.e., a map that sends each simplex of the input complex $\cI$ to a subcomplex of the output complex $\cO$.
For a given facet $\sigma=\{(1,v_1),\dots,(n,v_n)\}\in\cI$, the set of facets in $\Delta(\sigma)$ is defined by 
\begin{equation}
\label{eq:defDelta}
\{(1,v'_1),\dots,(n,v'_n)\} \in\Delta(\sigma) \iff  (v'_1,\dots,v'_n)\in F(v_1,\dots,v_n).
\end{equation}
The carrier map $\Delta$ of binary consensus maps every input $(n-1)$-facet~$\sigma$ containing both input values~0 and~1 to the two $(n-1)$-facets $\tau_0$ and $\tau_1$, and maps each $(n-1)$-facet $\sigma_b$ with a unique input value $b\in\{0,1\}$ to the output $(n-1)$-facet~$\tau_b$. 

\begin{figure}[tb]
\centerline{\includegraphics[width=16cm]{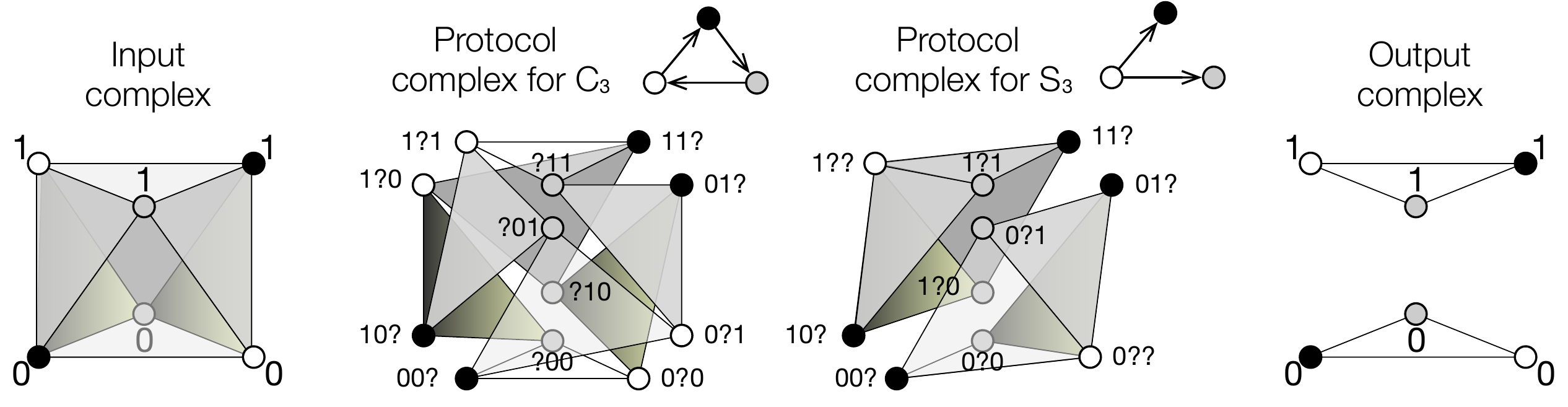}}
\caption{\sl\small Impact of the information-flow graph on the protocol complex for binary consensus with three processes. Labels next to vertices are input and output values, in the input and output complexes respectively, or views in protocol complexes. A view ``$xyz$'' labeling a vertex means that the process corresponding to this vertex knows the input values $x$ from process $\wcirc$, $y$ from process $\bcirc$, and $z$ from process~$\gcirc$. A question mark in a label indicates that the process does not know the corresponding input value. }
\label{fig:scisor}
\end{figure}

In any distributed computing model, in each point in time during the execution of a communication protocol, one can define a complex whose vertices are pairs $(p,w)$ where $w$ is the state of process $p$, i.e., its \emph{view} of the computation. 
A set of vertices with distinct process names forms a \emph{protocol simplex} if there is a protocol execution where those processes collect those views. All possible protocol simplexes make up the \emph{protocol complex}. 

In the \knowall{} model, we can identify the view of each process with the set of inputs it is aware of.
Due to synchrony and absence of failures, in an instance of the \knowall{} model with information-flow graph $\kagraphr$, each input simplex in $\cI$ results in a single execution, that leads to a unique simplex in the protocol complex $\cP$ associated with $\kagraphr$. 
The following fact is a direct consequence of the definition of the information-flow graph.

\begin{lemma}
\label{lem-protocol-complex}
Let $\kagraph$ be an instance of the \knowall\/ model, let $r\geq 0$, and let $(\cI,\cO,\Delta)$ be a task. The protocol complex $\cP$ associated with $\kagraph$ and $\cI$ after $r$ rounds, is the pure $(n-1)$-dimensional complex whose facets are all the sets of the form $\{(1,w_1),\dots,(n,w_n)\}$ 
for which there exists a facet $\{ (1,v_1),\dots,(n,v_n)\}$ of $\cI$ such that,
for $i=1,\dots,n$, the view $w_i$ satisfies
\[
w_i=\{(j,v_j) : i=j \;\mbox{or} \; (j,i)\in E(\kagraphr)\}.
\] 
The function~$\Xi:\cI\to\cP$ maps each facet of $\cI$ to a single facet of $\cP$, satisfying 
\[
\Xi(\{(1,v_1),\dots,(n,v_n)\}) =\{(1,w_1),\dots,(n,w_n)\}.
\]
\end{lemma} 

\paragraph{Notation.}
Given a view $w=\{(j_1,v_{j_1}),\dots,(j_k,v_{j_k}) \}$, we denote by $\name(w)=\{j_1,\dots,j_k\}$ the set of processes whose names appear in $w$, and by $\val(w)=\{v_{j_1},\dots,v_{j_k}\}$ the multiset of values appearing in $w$. 

An important point stated in Lemma~\ref{lem-protocol-complex} is that the facets of the input complex $\cI$ are preserved in the protocol complex $\cP$, i.e., there is a one-to-one correspondence between the facets of these two complexes. 
This phenomena is an important and interesting consequence of the fact that the system we study is not prone to asynchrony or failures, making the computation deterministic. However, while the facets of $\cI$ and $\cP$ remain in a one-to-one correspondence, their intersections may not possess  this property: two facets intersecting on a face of some dimension in $\cI$ may correspond to two facets intersecting only in a face of a lower dimension in $\cP$ (but not of a higher dimension). 
We thus use the terminology ``scissor cuts'' for describing  the transformation of $\cI$ into $\cP$ by $\Xi$.

\paragraph{Terminology.} 
	In the \knowall\/ model, the spreading map $\Xi:\cI\to\cP$ can be interpreted as a consequence of a sequence of topological deformations, which we call \emph{scissor cuts}.
	This mapping is a special type of spreading map, defined for colored complexes. Specifically, $\Xi$ maps each facet of $\cI$ to a single facet of $\cP$,  and it maps every simplex $\sigma\in \cI$ to the set  
	\[
	\Xi(\sigma)=\{\tau \in \Xi(\phi) : \; \mbox{$\phi$ is a facet of $\cI$}, \; \sigma\subseteq\phi\mbox{, and}\; \name(\tau)=\name(\sigma)\}. 
	\]

\paragraph{Example.}
Figure~\ref{fig:scisor} displays  two illustrations of the protocol complex for binary consensus, for two different information-flow graphs on three processes: the consistently directed cycle~$\kagraphr=C_3$, and the directed star~$\kagraphr=S_3$ whose center has two out-neighbors. Process names are omitted, and instead are represented by the  colors of the circles ($\wcirc$, $\gcirc$, and $\bcirc$). The number of vertices in the protocol complexes depends on the information-flow graph. Let us focus first on process~$\wcirc$. A vertex $(\wcirc,v)$ in the input complex yields two vertices in the protocol complex for $C_3$, depending on the input value received from process~$\gcirc$. Instead, a vertex $(\wcirc,v)$ in the input complex yields a single vertex in the protocol complex for $S_3$ because, according to this information-flow graph, process~$\wcirc$ receives no inputs from other processes. On the other hand, every vertex $(\bcirc,v)$ in the input complex yields two vertices in both protocol complexes. This is because, in both information-flow graphs, $C_3$ and $S_3$, process~$\bcirc$ receives the input from process~$\wcirc$. Similarly, every vertex $(\gcirc,v)$ in the input complex yields two vertices in both protocol complexes, because, in both information-flow graphs, process~$\gcirc$ receives the input from another process, from process~$\bcirc$ in $C_3$ and from process~$\wcirc$ in~$S_3$. 

\medskip

An important property of scissor cuts is that they are not arbitrary, but preserve a lot of the structure of the input complex. In fact, they can be understood by considering only cuts occurring between adjacent facets. To get an intuition of this property, let us consider the input complex $\cI$ on Figure~\ref{fig:scisor}, and let $\sigma_0=\{(\wcirc,0),(\gcirc,0),(\bcirc,0)\}$. Let $\sigma_1,\sigma_2,\sigma_3$ be the three facets of $\cI$ of the form $\{(\wcirc,x),(\gcirc,y),(\bcirc,z)\}$ with $x+y+z =1 $, and let $\sigma_4,\sigma_5,\sigma_6$ be the three facets of $\cI$ of the form $\{(\wcirc,x),(\gcirc,y),(\bcirc,z)\}$ with $x+y+z = 2$. Note that, for $1\leq i \leq 3$, $\sigma_i$ shares an edge with $\sigma_0$, while, for $4\leq i \leq 6$,  $\sigma_i$ only shares a vertex with $\sigma_0$. In fact, we have $ \bigcup^6_{i=1} ( \sigma_0  \cap \sigma_i ) = \bigcup^3_{i=1} (\sigma_0\cap \sigma_i )$, which is the empty triangle formed by the edges and nodes of $\sigma_0$. Now, let us consider the protocol complexes for $C_3$ and for $S_3$. In both cases, the cuts preserve the latter equality, that is, we also have  $\bigcup^6_{i=1}\big(\Xi(\sigma_0) \cap  \Xi(\sigma_i) \big)  = \bigcup^3_{i=1} \big( \Xi(\sigma_0) \cap \Xi(\sigma_i) \big)$. This property is not coincidental, but systematically holds, as shown below. 

\begin{lemma}\label{lem:first-of-two}
Let $1\leq s\leq t$ be integers. Let $\sigma_0,\sigma_1,\dots,\sigma_t$ be different facets of $\cI$ such that, for every $i\in\{1,\dots,s\}$, $\sigma_i$ is sharing an $(n-2)$-face with~$\sigma_0$. 
If $\bigcup^t_{i=1} (\sigma_0 \cap \sigma_i ) = \bigcup^s_{i=1} (\sigma_0 \cap \sigma_i)$
then $\bigcup^t_{i=1} \big(\Xi(\sigma_0) \cap \Xi(\sigma_i) \big)  = \bigcup^s_{i=1} \big( \Xi(\sigma_0) \cap \Xi(\sigma_i) \big).$
\end{lemma}

\begin{proof}
It is sufficient to show that, for every $s+1\leq i \leq t$, $\Xi(\sigma_0) \cap \Xi(\sigma_i) \subseteq \bigcup^s_{j=1} \big( \Xi(\sigma_0) \cap \Xi(\sigma_j) \big)$ whenever $\Xi(\sigma_0) \cap \Xi(\sigma_i) \neq \emptyset$. So, let us consider a facet $\sigma_i\in\cI$ such that $\Xi(\sigma_0) \cap \Xi(\sigma_i) \neq \emptyset$, and let $(p_j, w_j) \in \Xi(\sigma_0) \cap \Xi(\sigma_i)$. Let $p_k\in\name(w_j)$, i.e., $p_k$ is either an in-neighbor of $p_j$ in $\kagraphr$, or $p_k=p_j$. Let $x_k$ be the input of $p_k$, so $(p_k,x_k)\in w_j$. Since $(p_j, w_j) \in \Xi(\sigma_0) \cap \Xi(\sigma_i)$, it holds that $(p_k,x_k)\in \sigma_0 \cap \sigma_i$. By assumption, $\sigma_0 \cap \sigma_i \subseteq \bigcup^s_{j=1} (\sigma_0 \cap \sigma_j)$, so there exists $\ell\in\{1,\dots,s\}$ such that $\sigma_0 \cap \sigma_i \subseteq \sigma_0 \cap \sigma_\ell$. It follows that $(p_k,x_k)\in \sigma_0 \cap \sigma_\ell$. As a consequence,  $(p_k,x_k)$ is in the view of $p_j$ in $\Xi(\sigma_\ell)$. Therefore, $(p_j, w_j) \in \Xi(\sigma_0) \cap \Xi(\sigma_\ell)$. We conclude that $\Xi(\sigma_0) \cap \Xi(\sigma_i) \subseteq  \Xi(\sigma_0) \cap \Xi(\sigma_\ell)$, as desired. 
\end{proof}

\subsection{Topological Characterization of Task Solvability}

So far, we have proceeded in two parallel paths.
The first, \emph{operational} path,
was about algorithms in the \knowall\/ model, 
where information propagates between processes
according to some information-flow pattern $\kagraphr$ (cf. Section~\ref{sec:model}). 
The second, \emph{topological} path,
relates the inputs of processes defined by an input complex,
their views modeled in the protocol complex,
and their desired outputs, 
appearing in the output complex (cf. Section~\ref{subsec:topologicalcomputing}). 
The connection between these paths is established in the next lemma, which directly follows from the definitions. 

\begin{lemma}\label{theo:iffcondition4solvability}
A task $(I,O,F)$ is solvable in $r$ rounds  in the \knowall{} model $\kagraph$ if and only if, 
for the topological formulation $ (\cI,\cO,\Delta)$ of the task, 
there exists a chromatic simplicial map $\delta:\cP\rightarrow\cO$
satisfying $\delta(\Xi(\sigma))\in\Delta(\sigma)$  for every facet $\sigma\in\cI$, 
where $\cP$ is the protocol complex associated with $\kagraph$ and $\cI$ after $r$ rounds. 
\end{lemma}

The simplicial map $\delta:\cP\rightarrow\cO$ is called a \emph{decision map}. If $\delta(i,w_i)=(i,v'_i)$, then the corresponding algorithm specifies that process~$i$ with view~$w_i$ outputs $f(i,w_i)=v'_i$. 
 
\paragraph{Example.}

Let us consider Figure~\ref{fig:scisor} again. The protocol complex for $S_3$ is disconnected, while for $C_3$ it is 0-connected (i.e., path-connected). The presence of a universal node~$\wcirc$ (dominating all other nodes) in the information-flow graph~$S_3$ results in all processes becoming aware of the input of the process corresponding to that node.  Therefore, the protocol complex for $S_3$ is split into two subcomplexes, the one corresponding to process~$\wcirc$ with input~0, and the other corresponding to process~$\wcirc$ with input~1. Similarly, the protocol complex for the complete graph $K_3$ with bidirectional edges (not depicted in Figure~\ref{fig:scisor}), is entirely disconnected, i.e., composed of eight pairwise non-intersecting facets, because there is no uncertainty under the complete information-flow graph, as every process receives the input of every other process.  

Since the protocol complex for $S_3$ is disconnected, consensus is solvable in this graph. To see why, consider $\delta$ that maps every vertex $(p,0**)$ of the protocol complex to vertex $(p,0)$ of the output complex, and every vertex $(p,1**)$ of the protocol complex to vertex $(p,1)$ of the output complex. This is a chromatic simplicial map, and thus, by Lemma~\ref{theo:iffcondition4solvability} consensus is solvable. In contrast, there is no such mapping $\delta:\cP\to\cO$ for the protocol complex $\cP$ corresponding to~$C_3$, because $\cP$  is 0-connected. Let us consider the path $((\wcirc,1?1),(\gcirc,?01),(\bcirc,00?))$ in the protocol complex for $C_3$. Vertex $(\wcirc,1?1)$ must be mapped to vertex $(\wcirc,1)$ in the output complex because $(\wcirc,1?1)$ belongs to a facet with all processes having input value~1. Similarly, vertex $(\bcirc,00?)$ must be mapped to vertex $(\bcirc,0)$ because $(\wcirc,00?)$ belongs to a facet with all processes having input value~0. If a mapping $\delta$ maps $(\gcirc,?01)$ to $(\gcirc,1)$, then the simplex $\{(\gcirc,?01),(\bcirc,00?)\}$ is mapped to  $\{(\gcirc,1),(\bcirc,0)\}$, which is not a simplex of~$\cO$. The same occurs if $(\gcirc,?01)$ is mapped to $(\gcirc,0)$, as $\{(\wcirc,1),(\gcirc,0)\}$ is not a simplex of~$\cO$. Thus, there is no simplicial map~$\delta$, and, by Lemma~\ref{theo:iffcondition4solvability}, consensus is not solvable. 

\section{Connectivity and Domination}

In this section, we establish a connection between a given instance of the \knowall\/ model, the structure of the information-flow graph resulting from this instance, and the topology of the protocol complex induced by the information-flow graph.
In particular, we show that, assuming that the input complex $\cI$ is a pseudosphere (which is the case, e.g., for consensus and $k$-set agreement), if the information-flow graph has large domination number, then the protocol complex is highly connected. Later in the paper, we show that high connectivity is an obstacle to solving agreement tasks in the \knowall{} model, as was shown in the past for several other models of computation.

Recall that a \emph{dominating set} of a directed graph~$G$ is a set of nodes $D\subseteq V(G)$ such that, for every node~$v\in V(G)\setminus D$, there exists a node $u\in D$ such that $(u,v)\in E(G)$. The \emph{domination number} of a digraph~$G$, denoted $\gamma(G)$, is the minimum~$k$ such that $G$ has a dominating set of size~$k$. Also recall that, for $k\geq 1$, a complex is $k$-connected if for every $1 \leq k' \leq k$, any continuous map from the $k'$-dimensional sphere to a geometric realization of the complex can be extended to a continuous map from the $(k'+1)$-dimensional disk.

We now state and prove one of our main technical contributions. 

\begin{theorem} \label{thm:domination-implies-connectivity}
	Let $\kagraph$ be an instance of the \knowall{} model, $r$ be a positive integer, and $(\cI,\cO,\Delta)$ be a task. 
	If $\gamma(\kagraphr)> k$, then the protocol complex~$\cP$ associated with $\kagraph$ and $\cI$ after $r$ rounds is at least $(k-1)$-connected.  
\end{theorem}

The rest of this section is dedicated to the proof of Theorem~\ref{thm:domination-implies-connectivity}. This proof exemplifies the nature of the scissor cuts produced by $\Xi$ on $\cI$, resulting in $\cP$. 

Let $I$ be the set of input values of the task $(\cI,\cO,\Delta)$. 
If $|I| = 1$, then Theorem~\ref{thm:domination-implies-connectivity} holds as the protocol complex~$\cP$ is a single simplex of dimension $n-1$.
From this point on, we therefore assume that $|I| > 1$.

We begin by presenting another important property satisfied by the scissor cuts. It relates the ``depth of the cuts'' applied to $\cI$ with the connectivity of the resulting complex $\cP$. Later we will see that the domination number of the information-flow graph is precisely the ``depth of the cut''.
To get an intuition of the property, let us again consider Figure~\ref{fig:scisor}. Let $\sigma_0=\{(\wcirc,0),(\gcirc,0),(\bcirc,0)\}\in\cI$, and let $\sigma_1=\{(\wcirc,x),(\gcirc,y),(\bcirc,z)\}$ be one of the facets of $\cI$ with $x+y+z=1$. In the case of the protocol complex for $C_3$,  for every choice of $\sigma_1$, we observe that  $\Xi(\sigma_0)\cap \Xi(\sigma_1)$ is a single node, i.e., a simplex of dimension~0. On the other hand, in the case of the protocol complex for $S_3$, $\Xi(\sigma_0)\cap \Xi(\sigma_1)$ may either be an edge, or the empty set, depending on the choice of $\sigma_1$. The lemma below provides a lower bound on the dimension of $\Xi(\sigma_0)\cap \Xi(\sigma_1)$, as a function of the domination number of the information-flow graph.

\begin{lemma}\label{lem:second-of-two}
Let $t\geq 1$, and let $\sigma_0, \ldots, \sigma_t$ be a sequence of $t+1$ different facets of $\cI$ such that, for every $1\leq i\leq t$, $\sigma_i$ shares an $(n-2)$-face with $\sigma_0$. 
Then, for every $k$ such that $\gamma(\kagraphr)> k$, 
$
\bigcap^t_{i=0} \; \Xi(\sigma_i)
$
is of dimension at least $k-t$. 
\end{lemma}

\begin{proof}
Consider the simplex $S=\bigcap^t_{i=0} \; \Xi(\sigma_i)$
(this is a simplex since each $\Xi(\sigma_i)$ is a simplex, and so is their intersection).
In order to show that $S$ is of dimension at least $k-t$, we can restrict our attention to $t\leq k$,
as the dimension of a simplex cannot go below $-1$.

For every $i$, $1 \leq i \leq t$, let $p_i$ be the process whose identity does not appear in $\sigma_0 \cap \sigma_i$, i.e., $p_i\notin\name(\sigma_0 \cap \sigma_i)$.
Since $\sigma_0$ and $\sigma_i$ share an $(n-2)$-face, $p_i$ is uniquely defined. 

First, we show that the set 
\[
D = \text{names}(S) \cup \{p_1,\ldots,p_t\}
\]
is a dominating set of $\kagraphr$.
Let $q$ be a process in $[n]\setminus D$ (if no such process exists, $D=[n]$, which is trivially a dominating set).
Since, in particular, $q\notin \text{names}(S)$, we get that $q$ has different views in the executions $\Xi(\sigma_0)$ and $\Xi(\sigma_i)$, for some $i\in\{1,\dots,t\}$. 
As $\sigma_0$ and $\sigma_i$ share an $(n-2)$-face, only the process $p_i$ is able to distinguish between the two corresponding input configurations.
Hence, there is a directed edge from $p_i$ to $q$ in $\kagraphr$.
Therefore, $D$ is a dominating set for $\kagraphr$. 

The fact that $D$ is a dominating set for $\kagraphr$  implies that 
\[
|S| + |\{p_1,\ldots,p_t\}| \geq |D|\geq k+1.
\] 
It follows that $|S| \geq k + 1 - |\{p_1,\ldots,p_t\}| \geq k + 1 - t $, from which we conclude that the dimension of $S$ is at least $k-t$. 
\end{proof}

Lemmas~\ref{lem:first-of-two} and~\ref{lem:second-of-two} are actually the key facts enabling us to prove that if the domination number of $\kagraphr$ is large, then the protocol complex after $r$ rounds has high connectivity.

Recall that the $d$-\emph{skeleton} $\skel^{d}\cK$ of a complex~$\cK$ is the complex composed of all the faces of $\cK$ of dimension at most $d$.
The proof also uses the concept of \emph{shellable} complexes, defined as follows.

\begin{definition} 
Let $\mathcal{K}$ be a pure complex of dimension~$d$. 
$\mathcal{K}$ is \emph{shellable} if there is an ordering 
$\phi_1, \ldots, \phi_r$ of its facets such that, for every $1 \leq t \leq r-1$, the complex 
$
\cK_t= \big( \bigcup^t_{i=1} \phi_i \big) \cap \phi_{t+1}
$
is a pure subcomplex of dimension $d-1$ of $\skel^{d-1}\phi_{t+1}$. 
Such a sequence $\phi_1, \ldots, \phi_r$  of facets is called a \emph{shelling order} of $\mathcal{K}$. 
\end{definition}

One simple example of a shellable complex is that of a complex defined by the boundary of a single simplex, or any subcomplex of such complex; in this case, any order of the facets is a shelling order, as formally stated next.

\begin{lemma}[Theorem 13.2.2 in {\cite{HerlihyKR13}}]
	\label{fact-shelling-boundary}
	Let $\sigma$ be a simplex of dimension $d$, and let $\mathcal{K}$ be a pure $(d-1)$-dimensional subcomplex of $\skel^{d-1}\sigma$.
	Then $\mathcal{K}$ is shellable, and any sequence of its facets is a shelling order of $\mathcal{K}$.
\end{lemma}

Another typical example of a shellable complex is a pseudosphere.

\begin{lemma}[Theorem 13.3.6 in \cite{HerlihyKR13}]
\label{fact-pseudosphere-shellable}
The pseudosphere $\Psi(\{1,\ldots,n\},I)$ is shellable.
\end{lemma}

To get an intuition of why a pseudosphere is shellable, let us consider the input complex displayed on Figure~\ref{fig:scisor}, which is the pseudosphere $\Psi(\{1,2,3\},\{0,1\})$. 
Set $\phi_1$ as the lower facet (triangle) $\{(\wcirc,0),(\gcirc,0),(\bcirc,0)\}$. 
Set $\phi_2$, $\phi_3$, and $\phi_4$ as the three facets $\{(\wcirc,0),(\gcirc,0),(\bcirc,1)\}$, $\{(\wcirc,0),(\gcirc,1),(\bcirc,0)\}$, and $\{(\wcirc,1),(\gcirc,0),(\bcirc,0)\}$, in arbitrary order. 
Then set $\phi_5$, $\phi_6$, and $\phi_7$ as the three facets $\{(\wcirc,0),(\gcirc,1),(\bcirc,1)\}$, $\{(\wcirc,1),(\gcirc,1),(\bcirc,0)\}$, and $\{(\wcirc,1),(\gcirc,0),(\bcirc,1)\}$, in arbitrary order. 
Finally, set $\phi_8$ as $\{(\wcirc,1),(\gcirc,1),(\bcirc,1)\}$. 
This is one of the many shelling orderings establishing that $\Psi(\{1,2,3\},\{0,1\})$ is shellable. However, note that picking $\phi_1=\{(\wcirc,0),(\gcirc,0),(\bcirc,0)\}$ and $\phi_2=\{(\wcirc,0),(\gcirc,1),(\bcirc,1)\}$ cannot be a part of any shelling order of $\Psi(\{1,2,3\},\{0,1\})$, neither can the choice $\phi_1=\{(\wcirc,0),(\gcirc,0),(\bcirc,0)\}$ and $\phi_2=\{(\wcirc,1),(\gcirc,1),(\bcirc,1)\}$. 
Lemma~\ref{fact-pseudosphere-shellable} implies that the input complex $\cI = \Psi(\{1,\ldots,n\},I)$ is a pure $(n-1)$-dimensional complex that is shellable. 

The lemma below is the core of the technical part of the proof of Theorem~\ref{thm:domination-implies-connectivity}. It can be thought as a finer version of Lemma 13.4.2 in~\cite{HerlihyKR13}.

\begin{lemma}
	\label{lemma-shellability-connectivity}
	Let $\mathcal{A}$ be a pure shellable complex of dimension $d$, and let $\mathcal{B}$ be any complex.
	Let  $f$ be an onto mapping from the facets of $\mathcal{A}$ to the facets of $\mathcal{B}$. Fix an integer $\ell\geq 0$ such that $f$ satisfies the following two properties: 
	\begin{description}
	\item[\small P1:] 
		For every two integers $s,t$ satisfying $1\leq s\leq t$, and every sequence $\sigma_0,\sigma_1,\dots,\sigma_t$ of $t+1$ different facets of $\cA$ where, for every $i\in[s]$, $\sigma_i$ shares a $(d-1)$-face with~$\sigma_0$ and $\bigcup^t_{i=1} (\sigma_0 \cap \sigma_i) = \bigcup^s_{i=1} (\sigma_0 \cap \sigma_i)$, it holds that 
		$
		  \bigcup^t_{i=1} \big (  f(\sigma_0) \cap f(\sigma_i) \big)  = \bigcup^s_{i=1} \big (f(\sigma_0) \cap f(\sigma_i) \big);
		$	
	\item[\small P2:]  
		For every $t\geq 0$, and for every sequence $\sigma_0,\sigma_1, \ldots, \sigma_t$ of $t+1$ different facets of $\cA$
		where each $\sigma_i$, $i\geq 1$, shares a $(d-1)$-face with $\sigma_0$, the complex
		$
		\bigcap^t_{i=0} \; f(\sigma_i)
		$
		is of dimension at least $\ell-t+1$. 
	\end{description}
	Then, $\mathcal{B}$ is $\ell$-connected.
\end{lemma}

Observe that the two properties assumed in the statement of Lemma~\ref{lemma-shellability-connectivity} hold for the scissor cuts, where $\cA=\cI$, $\cB=\cP$, and $f=\Xi$, as established in lemmas~\ref{lem:first-of-two} and~\ref{lem:second-of-two}. The proof of Lemma~\ref{lemma-shellability-connectivity} uses two known facts from combinatorial topology. First, a subcomplex of a boundary complex of a simplex is shellable (Lemma~\ref{fact-shelling-boundary} above). Second, in order to show that a complex is well connected, we can present a cover of it by two subcomplexes, such that they both have high connectivity, and so does their intersection. 
This latter fact is a simple corollary of the so-called \emph{Nerve Lemma}, and is stated next.

\begin{lemma}[Corollary 10.4.3 in \cite{HerlihyKR13}]
	\label{lemma-nerve}
	Let $\mathcal{K}$ and $\mathcal{L}$ be two $\ell$-connected complexes such that $\mathcal{K} \cap \mathcal{L}$ is $(\ell-1)$-connected.
	Then, $\mathcal{K} \cup \mathcal{L}$ is $\ell$-connected.
\end{lemma}

In order to prove Lemma~\ref{lemma-shellability-connectivity}, we need one more notation, of \emph{petals},
which is used only in the proof of this lemma.
Given a shelling order $\phi_1, \ldots, \phi_r$ of a shellable $d$-dimensional complex $\cK$, the complex
$\cK_t=\left(\bigcup^t_{i=1} \phi_i\right) \cap \phi_{t+1}=\bigcup^t_{i=1} (\phi_i \cap \phi_{t+1})$ is a subcomplex of the boundary complex of $\phi_{t+1}$, and as such, it can be represented as a union of the complexes induced
by some $(d-1)$-faces $\tau_1, \hdots, \tau_s$ of $\phi_{t+1}$. 
Moreover, each $\tau_j$, $1\leq j \leq s$, is a face of a facet $\phi_{i_j}$ of $\bigcup^t_{i=1} \phi_i$ for some $1\leq i_j\leq t$, i.e., $\phi_{t+1}$ and $\phi_{i_j}$ share the $(d-1)$-face $\tau_j$. In short, 
\[\cK_t= \bigcup^t_{i=1} (\phi_i  \cap \phi_{t+1}) = \bigcup^s_{j=1}  \tau_j = \bigcup^s_{j=1} ( \phi_{i_j}  \cap \phi_{t+1}).\]
We call the set $S=\{i_1,\dots,i_s\}$ the \emph{petals} of $\phi_{t+1}$. More precisely: 
\begin{definition}\label{def:petals}
		Given a shelling order $\phi_1, \ldots, \phi_r$ of a pure shellable complex $\cK$ of dimension~$d$, and an index $t\in\{1,\dots,r-1\}$, a \emph{petal of $\phi_{t+1}$} is a facet $\phi_i$, $1\leq i\leq t$, for which $\phi_i\cap\phi_{t+1}$ is of dimension $d-1$.
		A \emph{minimal covering set of petals of $\phi_{t+1}$} is a set $P\subseteq \{1,\dots,t\}$ of indices that have the following properties.
		\begin{enumerate}
			\item
			Covering: $\bigcup_{i\in P} ( \phi_{i}  \cap \phi_{t+1}) = \bigcup^t_{i=1} (\phi_i  \cap \phi_{t+1})$.
			\item 
			Petals: For each $i\in P$, $\phi_i\cap\phi_{t+1}$ is a petal.
			\item
			Minimal: For each $i,i'\in P$ such that  $i\neq i'$, $\phi_i\cap\phi_{t+1}\neq \phi_{i'}\cap\phi_{t+1}$.	
		\end{enumerate}
\end{definition}

We now have all the ingredients to prove Lemma~\ref{lemma-shellability-connectivity}.

\paragraph{Proof of Lemma~\ref{lemma-shellability-connectivity}.~}
	The proof is by induction on $\ell$. For each $\ell$, we will apply induction on the length $m$ of a shelling order of $\cA$.
	
	\medskip
	
	\noindent -- \; For $\ell=0$, we need to prove that $\cB$ is $0$-connected. We do so 
	by induction on the length of a shelling order of $\cA$.
	So, let us fix a shelling order $\phi_1, \ldots, \phi_m$ of $\cA$. 
	We have $\cB = \bigcup^m_{i=1} f(\phi_i)$ because $f$ is an onto mapping from the facets of $\cA$ to the facets of $\cB$.
	
	\begin{itemize} 
	\item The base case is $m=1$, that is, $\cB=f(\phi_1)$. In this case, $\cB$ is a simplex. Thus, it is $c$-connected for every $c$, and  in particular it is 0-connected. 
	\item
	For the inductive step, with $m\geq 2$, let us assume that properties P1 and P2 hold for $\cA,\cB,f$ and $\ell=0$, and that the lemma holds for $m-1$.
	Next, 
	we apply Lemma~\ref{lemma-nerve} in order to prove that 
	$\cB=\big (\bigcup^{m-1}_{i=1}f(\phi_i) \big) \cup f(\phi_m)$ is $0$-connected.
	To this end, we prove 
	that (1) $\bigcup^{m-1}_{i=1} f(\phi_i)$ is $0$-connected,
	(2) $f(\phi_m)$ is $0$-connected,
	and (3) their intersection $\big( \bigcup^{m-1}_{i=1} f(\phi_i)  \big) \cap f(\phi_m)$ is $(-1)$-connected.

	Consider the subcomplex $\bigcup^{m-1}_{i=1} \phi_i$ of $\cA$, the subcomplex $\bigcup^{m-1}_{i=1} f(\phi_i)$ of $\cB$, and the mapping $f$ between their facets.
	Note that since properties P1 and P2 hold for  $\cA,\cB,f$, and $\ell=0$, they also hold for the subcomplexes defined above. 
	Since the lemma holds for $m-1$, we conclude that $\bigcup^{m-1}_{i=1} f(\phi_i)$ is $0$-connected.
	We also have that $f(\phi_m)$ is $0$-connected, as it is a simplex.	
	
	For the intersection, let $\phi_j$, $1\leq j\leq m-1$, be a petal of $\phi_m$, i.e., $\phi_j\cap\phi_{m}$ is of dimension $d-1$. 
	Such a~$\phi_j$ exists since the sequence $\phi_1,\ldots,\phi_m$ is a shelling order.
	By Property P2 with $\sigma_0=\phi_m$, $\sigma_1=\phi_j$, and $t=1$, the complex 
	$f(\phi_j) \cap f(\phi_m)$ is of dimension at least~$0$, and, specifically, it is non-empty.
	This implies that 
	$\bigcup^{m-1}_{i=1} \big ( f(\phi_i)  \cap f(\phi_m)\big )$ is also non-empty, i.e., $(-1)$-connected, as desired.
	\end{itemize}
	
\noindent -- \; For the inductive step, from $\ell-1$ to $\ell$, where $\ell\geq 1$, let us assume that the statement of the theorem holds for $\ell-1$, and consider a shelling order $\phi_1, \ldots, \phi_m$ of $\cA$.
	Our aim is to show that ${\cB} = \bigcup^m_{i=1} f(\phi_i)$ is $\ell$-connected.
	As in the base case $\ell=0$, we proceed by induction on the length $m$  of the shelling order. 
	
	\begin{itemize} 
	\item 
	The base case is $m=1$.
	Again, since $f(\phi_1)$ is a simplex, it is $\ell$-connected.

	\item 
	For the inductive step, with $m\geq 2$, let us assume as before that properties P1 and P2 hold for $\cA,\cB,f$, and $\ell$, and that the lemma holds for $m-1$.
	We follow the same outline of proof, that is, we show that 
	(1) $\bigcup^{m-1}_{i=1} f(\phi_i)$ is $\ell$-connected,
	(2) $f(\phi_m)$ is $\ell$-connected,
	and (3) their intersection $\big( \bigcup^{m-1}_{i=1} f(\phi_i)  \big) \cap f(\phi_m)$ is $(\ell-1)$-connected.
	Lemma~\ref{lemma-nerve} then implies that 
	$\cB=\big (\bigcup^{m-1}_{i=1}f(\phi_i) \big) \cup f(\phi_m)$ is $\ell$-connected, as desired
	
	As in the case $\ell=0$, by considering the subcomplexes $\bigcup^{m-1}_{i=1} \phi_i$, $\bigcup^{m-1}_{i=1} f(\phi_i)$, and the mapping $f$ between their facets, and using the induction hypothesis for $\ell$ and $m-1$, we conclude that 
	$\bigcup^{m-1}_{i=1} f(\phi_i)$ is $\ell$-connected.
	We also have that $f(\phi_m)$ is $\ell$-connected, as it is a simplex.	
	
	We now show that $\left( \bigcup^{m-1}_{i=1} f(\phi_i) \right) \cap f(\phi_m)$ is $(\ell-1)$-connected.
	To this end, we define new simplicial complexes $\cA',\cB'$, and a mapping $f'$ between their facets.
	Let $P$ be a minimal covering set of petals of $\phi_m$. That is, 
	\[ 
	\bigcup^{m-1}_{i=1}( \phi_i  \cap \phi_m ) = 
	\bigcup_{i\in P} \left( \phi_i \cap \phi_m  \right),
	\]
	with each intersection $\phi_i \cap \phi_m$ on the right hand side ($i\in P$) being a distinct simplex of dimension $d-1$.
	By Property~P1, we have that  
	\[
	\bigcup^{m-1}_{i=1} \big (f(\phi_i)  \cap f(\phi_m) \big) = 
	\bigcup_{i\in P} \big( f(\phi_i) \cap f(\phi_m) \big).
	\]
	Let $\cA' = \bigcup_{i\in P} \left( \phi_i \cap \phi_m \right)$ and $\cB' = \bigcup_{i\in P} \left( f(\phi_i) \cap f(\phi_m) \right)$ be the right hand sides of the two above equalities. 
	Note that $\cA'$ is pure of dimension $d-1$ and is a subcomplex of the boundary complex of $\phi_m$, thus it is shellable by Lemma~\ref{fact-shelling-boundary}.
	By the last equality, it is sufficient to show that $\cB'$ is $(\ell-1)$-connected. 
	Consider the mapping $f'$ from the facets of $\cA'$ to those of $\cB'$
	defined by 
	\[f'(\phi_i \cap \phi_m) = f(\phi_i) \cap f(\phi_m)\]
	for every $i\in P$.
	This mapping is onto by the definition of $\cA'$ and $\cB'$.
	We now use Lemma~\ref{lemma-shellability-connectivity} for $\cA',\cB',f',d'=d-1$, and $\ell'=\ell-1$, which is applicable by the induction hypothesis.
	
	To show that Property~P1 holds for $\cA',\cB',f'$, and $d'$ (P1 does not depend on $\ell'$), 
	let us consider two integers $1\leq s'\leq t'$,
	and a sequence $\phi_{i_0} \cap \phi_m, \ldots, \phi_{i_{t'}} \cap \phi_m$
	of different facets of $\cA'$ where, for every $j\in[s']$, $\phi_{i_j} \cap \phi_m$ shares a $(d-2)$-face with~$\phi_{i_0} \cap \phi_m$, and 
	\[
	\bigcup^{t'}_{j=1} \big((\phi_{i_0} \cap \phi_m) \cap (\phi_{i_j} \cap \phi_m)\big) = \bigcup^{s'}_{j=1} \big((\phi_{i_0} \cap \phi_m) \cap (\phi_{i_j} \cap \phi_m)\big).
	\] 
	For each $j\geq1$, the fact that $i_0$ and $i_j$ are different petals of $\phi_m$ implies that $\phi_{i_0} \cap \phi_m\neq\phi_{i_j} \cap \phi_m$.
	Hence, the right hand side is a union of simplices of dimension exactly $d-2$, making it a pure complex of dimension $d-2$.

	We claim that $\{\phi_{i_1},\ldots,\phi_{i_{t'}}\}=\{\phi_{i_1},\ldots,\phi_{i_{s'}}\}$.
	Assume otherwise, i.e., there is a facet $\phi_{i'}$ with $i'\in\{i_1,\ldots,i_{t'}\}$ such that $\phi_{i'}\neq\phi_{i_j}$ for each $1\leq j\leq s'$.
	Since the complex in the right hand side is pure of dimension $d-2$, the equality implies that 
	$\phi_{i_0} \cap \phi_m \cap \phi_{i'}$ is a facet there, i.e., $\phi_{i_0} \cap \phi_m \cap \phi_{i'}=\phi_{i_0} \cap\phi_m \cap \phi_{i_j}$ for some $1\leq j\leq s'$.
	As $\phi_{i_0}, \phi_{i'}$, and $\phi_{i_j}$ are distinct petals of $\phi_m$, their three intersections with $\phi_m$, namely $\phi_{i_0}\cap\phi_m, \phi_{i'}\cap\phi_m$, and $\phi_{i_j}\cap\phi_m$, are distinct simplices of dimension $d-1$.
	The intersection of these three simplices must thus be of dimension at most $d-3$.
	However, the equality $\phi_{i_0} \cap \phi_m \cap \phi_{i'}=\phi_{i_0} \cap\phi_m \cap \phi_{i_j}$ implies that $\phi_{i_0} \cap \phi_m \cap \phi_{i'} \cap \phi_{i_j}=\phi_{i_0} \cap \phi_m \cap \phi_{i'}$, where the left hand side has dimension at most $d-3$ and the right hand side has dimension $d-2$, a contradiction.	
	Hence, $\{\phi_{i_1},\ldots,\phi_{i_{t'}}\}=\{\phi_{i_1},\ldots,\phi_{i_{s'}}\}$, as claimed.
	
	Applying $f$ completes the proof of P1. Indeed, we have 	
	\[
	\bigcup^{t'}_{j=1}f(\phi_{i_j}) = 
	\bigcup^{s'}_{j=1}f(\phi_{i_j}),
	\]
	hence 
	\[
	f(\phi_{i_0}) \cap f(\phi_m) \cap \bigcup^{t'}_{j=1}f(\phi_{i_j}) = 
	f(\phi_{i_0}) \cap f(\phi_m) \cap \bigcup^{s'}_{j=1}f(\phi_{i_j}),
	\]
	so 
	\[
	\bigcup^{t'}_{j=1}\big(f(\phi_{i_0}) \cap f(\phi_m) \cap f(\phi_{i_j}) \big)= 
	\bigcup^{s'}_{j=1}\big(f(\phi_{i_0}) \cap f(\phi_m) \cap f(\phi_{i_j}) \big),
	\]
	which implies	
	\[
	\bigcup^{t'}_{j=1} \big (  f'(\phi_{i_0} \cap \phi_m) \cap f'(\phi_{i_j} \cap \phi_m) \big) = \bigcup^{s'}_{j=1} \big (  f'(\phi_{i_0} \cap \phi_m) \cap f'(\phi_{i_j} \cap \phi_m) \big),
	\]
	completing the proof of Property P1.	
	
	To show that Property~P2 also holds for $\cA',\cB',f',d'$, and $\ell'$, let us fix $t'\geq 0$, and consider a sequence $\phi_{i_0} \cap \phi_m, \ldots, \phi_{i_{t'}} \cap \phi_m$ of $t'+1$ facets of $\cA'$, i.e., $i_j\in P$ for every $j=0,\dots,t'$.
	For each $1\leq j\leq t'$, the facet $\phi_{i_j} \cap \phi_m$ shares a $(d-2)$-face with $\phi_{i_0} \cap \phi_m$ (this is the assertion in P2, but in our specific case this is true for any choice of a sequence, since $\phi_{i_0}$ and $\phi_{i_j}$ are petals of $\phi_m$).
	We have 
	\[
	\bigcap^{t'}_{j=0} f'(\phi_{i_j}\cap \phi_m) = \bigcap^{t'}_{j=0} \big (f(\phi_{i_j}) \cap f(\phi_m)\big ) = \Big ( \bigcap^{t'}_{j=0} f(\phi_{i_j}) \Big ) \cap  f(\phi_m).
	\]
	To show that this intersection is of dimension at least $\ell'-t'+1$, we use the assumption that Property~P2 holds in the original setting, with $\cA,\cB,f,d$, and $\ell$. 
	We apply P2 with $t=t'+1$, and with the sequence $\phi_m, \phi_{i_0},\ldots,\phi_{i_{t'}}$ of $t+1$ facets of $\cA$, where $\phi_m$ is the facet intersecting all the other $t$ facets.
	Note that, for every $0\leq j \leq t'$, $\phi_{i_j}$ shares a $(d-1)$-face with~$\phi_m$ since $\phi_{i_j}$ is a petal of $\phi_m$.
	Property~P2 then assures that the above intersection is of dimension at least $\ell-t+1 =\ell'-t'+1$.
	As a consequence, Property~P2 also holds for $\cA',\cB',f',d'$, and $\ell'$.
	
	By induction, the theorem holds for $\cA',\cB',f',d'$, and $\ell-1$. 
	Therefore, $\cB'$ is $(\ell-1)$-connected. 
	It follows that $\bigcup^{m}_{i=1} f(\phi_i)$ is $\ell$-connected, which completes the induction step.
		\end{itemize} 
This completes the proof of Lemma~\ref{lemma-shellability-connectivity}. \qed

To complete the proof of Theorem~\ref{thm:domination-implies-connectivity}, observe that lemmas~\ref{lem:first-of-two} and~\ref{lem:second-of-two} guarantee that the two assumptions in the statement of Lemma~\ref{lemma-shellability-connectivity} hold for $\ell=k-1$, and  thus Lemma~\ref{lemma-shellability-connectivity} shows that $\cP$ is $(k-1)$-connected, as desired. 

\section{Applications to Agreement Tasks}

In this section, we show how to use topology to derive lower bounds and impossibility results in the context of distributed network computing, with implications to classical models such as the \local\/ model and dynamic networks. 
We start by studying the classical agreement task of consensus, and a relaxed problem called $k$-set agreement, where processes may decide on up to $k$ different values. We then move our attention to a different variant of consensus, \emph{approximate agreement}, where decision values must all reside in a small range.

\subsection{Consensus and Set-Agreement}
\label{subsec:appli2consensus}

Let $k\geq 1$. Recall that, in the $k$-set agreement task, the processes must agree on at most~$k$ of the input values. 
It is known that, in the context of asynchronous shared memory computing, the level of connectivity of the protocol complex is closely related to the
ability to solve $k$-set agreement (see, e.g.,~\cite{HerlihyR94,HerlihyR00,HerlihyS99}). 
We show that this also holds in the \knowall\/ model. 

To illustrate this  result, let us assume that the protocol complex $\cP$ is 0-connected (path connected), i.e., for every pair of vertices in $\cP$ there is a sequence of edges of $\cP$ forming a path connecting these two vertices. Then, consensus, i.e., 1-set agreement, cannot be solved. (This is, e.g., the case of the protocol complex for $C_3$ in Figure~\ref{fig:scisor}, while the protocol complex for $S_3$ is not 0-connected.)
To see this, assume consensus can be solved, and let $(p,w)$ be a vertex in the protocol complex representing an execution for a process $p$ that decides~0,
and $(p',w')$ a vertex where $p'$ decides~1. Such vertices exist because in the execution with all inputs equal to $b \in \{0,1\}$, processes must decide $b$. The protocol complex~$\cP$  is 0-connected, so there is a path in~$\cP$ connecting $(p,w)$ and $(p',w')$. Each vertex in this path has a binary decision value, and so, some edge along that path must have one endpoint deciding~0 and the other endpoint deciding~1. This edge is in a facet whose outputs contain both 0s and 1s, contradicting the specification of consensus. 
Theorem~\ref{theo:impossibility-set-agreement} below uses a generalization of the above argument to higher-dimensional connectivity, in order to study the time complexity of solving $k$-set agreement in instances of the \knowall{} model.

More specifically, Theorem~\ref{theo:impossibility-set-agreement} states that, for any instance $\kagraph$ of the \knowall{} model, if the information-flow graph $\kagraphr$ associated with $\kagraph$ after $r$ rounds has a domination number $\gamma(\kagraphr)$ larger than $k$, then $k$-set agreement is not solvable in $\kagraph$ in $r$ rounds. It is not hard to prove this statement when the number of input values is at least the number of processes. Indeed, assume, for the purpose of contradiction, that there exists an algorithm~$A$ solving $k$-set agreement in $r$ rounds in $\kagraph$ with $\gamma(\kagraphr)>k$. 
Consider an input assignment $in$, where all input values $in(p)$, $p\in[n]$, are distinct. Let $S$ be the set of output values produced by $A$, and let $P$ be the set of processes with input values in $S$. We have $|P|=|S|\leq k<\gamma(\kagraphr)$, hence there is a process $q\notin P$ which does not hear from any node in $P$, and produces some output value~$v$. Let $p\in P$ be the process with input~$v$, and consider another input assignment~$in'$, where all inputs are the same, excepted for the input of $p$ which is replaced by some~$v'\neq v$. Since there is no edge $(p,q)$ in $\kagraphr$, the process~$q$ does not distinguish~$in'$ from~$in$, and outputs~$v$ in~$in'$ as well. This output is incorrect, as $v$ has not been proposed in~$in'$, and therefore $k$-set agreement is not solvable in $\kagraph$ in $r$ rounds. This reasoning cannot be applied when the number of input values is less than the number of processes, and in particular when the number of input values is $k+1<n$ (e.g., for binary consensus among at least three processes). Theorem~\ref{theo:impossibility-set-agreement} says that $k$-set agreement remains unsolvable even in this case. 

\begin{theorem}
	\label{theo:impossibility-set-agreement}
	Let $\kagraph$ be an instance of the \knowall{} model, and $r\geq 0$.
	If $\gamma(\kagraphr)> k$, then $k$-set agreement is not solvable  in $r$ rounds in $\kagraph$.
\end{theorem}

The proof of Theorem~\ref{theo:impossibility-set-agreement} uses Sperner's Lemma, which essentially states that every Sperner coloring with $d+1$ colors of a subdivision of a $d$-dimensional simplex contains a cell of that subdivision whose vertices are colored by $d+1$ distinct colors. The proof is also based on the notion of \emph{simplicial approximation} (see Chapter 3.7 in~\cite{HerlihyKR13}).

\begin{definition}
 Let $\cA$ and $\cB$ be simplicial complexes. 
 A carrier map $\chi: \cA \to 2^\cB$ has a \emph{simplicial approximation} if there exists a subdivision $\Div \cA$ of $\cA$, and a simplicial map ${g:\Div \cA \to \cB}$ such that, for every simplex $\sigma$ of $\cA$,  $g(\Div \sigma)$ is a subcomplex of $\chi(\sigma)$.  
\end{definition}

A crucial point is that high-connectivity implies the existence of simplicial approximation. 

\begin{lemma}[Theorem 3.7.7(2) in~\cite{HerlihyKR13}] \label{lem:simplicial-approx}
 Let $\cA$ and $\cB$ be simplicial complexes. Let ${\chi: \cA \to 2^\cB}$  be a carrier map such that, for every $d$-dimensional simplex $\sigma$ of $\cA$, the subcomplex $\chi(\sigma)$ is $(d-1)$-connected. Then $\chi$ has a simplicial approximation $(\Div,g)$.
\end{lemma}

\paragraph{Proof of Theorem~\ref{theo:impossibility-set-agreement}. }

By Theorem~\ref{thm:domination-implies-connectivity}, if $\gamma(\kagraphr)> k$ then the protocol complex~$\cP$ is at least $(k-1)$-connected. Thus, to establish Theorem~\ref{theo:impossibility-set-agreement} it is sufficient to show that if the protocol complex~$\cP$ for $\kagraphr$ is at least $(k-1)$-connected, then $k$-set agreement is not solvable in $r$ rounds in $\kagraph$. 
In fact, we show something slightly different: 
we first apply Theorem~\ref{thm:domination-implies-connectivity} to show that for every non-empty set $J$ of input values, the subcomplex of $\cP$ induced by $J$ is at least $(k-1)$-connected, and then prove that if all these subcomplexes are  $(k-1)$-connected then $k$-set agreement is not solvable.

Let us fix the information-flow graph $\kagraphr$ with domination number at least $k+1$. Let $I$ be a set of at least $k+1$ distinct values, and let $\cI$ denote the corresponding input complex. That is, $\cI$ is  the pseudosphere $\Psi(\{p_1,\ldots,p_n\},I)$, and $\cP$ denotes the protocol complex for the input complex~$\cI$ and the graph~$\kagraphr$. 
For any subset $J \subseteq I$, we are interested in the subcomplex of $\cP$ that arises when processes are given only inputs from $J$. 
Let $\cP[J]$ denote this subcomplex. 
The following claim strengthens the statement of Theorem~\ref{thm:domination-implies-connectivity}. 

\begin{claim}  \label{claim:sub-protocol-complex-are-connected}
	If $\kagraphr$  has domination number at least $k+1$, 
	then, for every non-empty $J \subseteq I$, $\cP[J]$  is at least $(k-1)$-connected.
\end{claim}

To establish the claim, let $\cI[J]$ be the input complex $\cI$ restricted to values in $J$.
Note that $\cI[J]$ is the pseudosphere  $\Psi(\{p_1,\ldots,p_n\},J)$, and, for a task where $\cI[J]$ is the input complex, $\cP[J]$ is the  corresponding protocol complex.
Theorem~\ref{thm:domination-implies-connectivity} immediately implies that $\cP[J]$ is at least $(k-1)$-connected, as claimed. 

For completing the proof of Theorem~\ref{theo:impossibility-set-agreement}, it is therefore sufficient to show the following. 

\begin{claim}
	\label{claim:connctivity-implies-no-set-agreement}
	If, for every non-empty $J \subseteq I$,
	$\cP[J]$ is at least $(k - 1)$-connected, 
	then $k$-set agreement is not solvable. 
\end{claim}

We proceed by proving this claim, inspired by a classical construction used for proving impossibility of $k$-set agreement in asynchronous, failure-prone distributed models (see, e.g.,~\cite[Theorem~{10.3.1}]{HerlihyKR13}).
The proof goes roughly as follows. We first observe that the existence of an algorithm for $k$-set agreement implies the existence of a simplicial map $f$ from the protocol complex $\cP$ to the boundary of the simplex $\sigma_I$ with vertex set~$I$. Second, we use the assumption that every subcomplex $\cP[J]$ is at least $(k-1)$-connected to derive the existence of a simplicial approximation $g: \Div\sigma_I \to \cP$ of $\chi$, where $\chi$ is the carrier map mapping every face $\sigma_J$ of $\sigma_I$ to the subcomplex $\cP[J]$. Combining $f$ and $g$ results in a simplicial map $h=f\circ g$ from $\Div\sigma_I$ to the boundary $\partial\sigma_I$ of $\sigma_I$. We then show that $h$ is a Sperner coloring of $\Div\sigma_I$ from which we derive a contradiction by applying Sperner's Lemma.  

Let $\sigma_I$ be the simplex of dimension $k$ with vertex set $I$, and let $\partial\sigma_I$ be its boundary complex,
that is, the complex of dimension $k-1$ whose facets correspond to the subsets of size $k$ of~$I$.
Each face of $\sigma_I$ corresponding to a set $J\subseteq I$ of values, and is denoted by $\sigma_J$. 
We emphasize that $\sigma_I$ and $\sigma_J$ are merely simplices with vertices labeled by the elements on $I$ and $J$, which we use in order to apply Sperner's Lemma. 
These simplices are not labeled with process names, and unlike the other simplicial complexes in this paper, they do not model the distributed system in hand.

Let us assume, for the purpose of contradiction, that $k$-set agreement is solvable. 
An algorithm for $k$-set agreement with set of input values $I$ implies a coloring of the vertices of $\cP$
with values in $I$ such that, for every simplex of $\cP$,
\begin{enumerate}
\item the set of colors assigned to the vertices of that simplex is of size at most $k$, and
\item for any set $J\subseteq I$ of colors, 
	all the nodes in $\cP[J]$ are colored with colors from $J$.
\end{enumerate}
The first property yields from the fact that processes must not output more that $k$ values in an execution, and the second property holds because processes may only decide on values that were proposed.
In other words, there exists a simplicial map 
\[f : \cP \to \partial\sigma_I\] such that, for every non-empty subset $J \subseteq I$, \[f(\cP[J]) \subseteq \partial\sigma_I \cap 2^J.\]
Let us define the carrier map 
\[
\chi:\sigma_I\to 2^\cP
\] 
by 
\[
\chi\left(\sigma_J\right) = \cP[J]
\]
for every $\sigma_J \subseteq \sigma_I$.
By the assumption of the claim, for every nonempty $\sigma_J \subseteq \sigma_I$, the subcomplex $\chi(\sigma_J)$ is $(k-1)$-connected. Therefore, by Lemma~\ref{lem:simplicial-approx}, the carrier map $\chi$ has a simplicial approximation,
that is, 
there exists a subdivision $\Div \sigma_I$ of $\sigma_I$, together with a simplicial map 
\[
g : \Div\sigma_I \to\cP,
\]
such that, for every simplex $\sigma_J \subseteq \sigma_I$, 
we have $g(\Div\sigma_J) \subseteq \chi(\sigma_J)$. 
Let us now consider the composition $h = f\circ g$ of simplicial maps. We have 
\[
h : \Div\sigma_I \to \partial\sigma_I.
\]
The map $h$ can be viewed as a coloring of $\Div\sigma_I$ with the vertices of $\partial\sigma_I$ such that each simplex in $\Div\sigma_I$ is colored with at most $k$ colors. Moreover, for each simplex $\sigma_J \subseteq \sigma_I$, we have
\[
h(\Div\sigma_J) = f(g(\Div\sigma_J)) \subseteq f(\chi(\sigma_J)) = f(\cP[J]).
\]
From the definition of $f$, it follows that each simplex in $\Div\sigma_J$ is colored by $h$ with values appearing in $\sigma_J$. 
Therefore, $h$ is a Sperner coloring of $\Div\sigma_I$. By Sperner's lemma, there exists a simplex $\tau$ of $\Div\sigma_I$, of dimension $k$, colored with all the $k+1$ colors. 
This is a contradiction, because $\tau$ is then mapped by $h$ to $\sigma_I$, which is not in the domain $\partial\sigma_I$ of $h$. This completes the proof of Claim~\ref{claim:connctivity-implies-no-set-agreement}, and of the theorem. 
\qed

Theorem~\ref{theo:impossibility-set-agreement} implies that, in particular, consensus solvability requires the information-flow graph to contain a universal node, i.e., a node that dominates all the other nodes. 

\paragraph{Remark.}
Observe that the converse of Theorem~\ref{theo:impossibility-set-agreement} also holds, i.e., if $\gamma(\kagraphr)\leq k$ then $k$-set agreement is solvable  in $r$ rounds in the \knowall\/ model~$\kagraph$. The algorithm performs as follows. Let $D$ be a dominating set for $\kagraphr$, with $|D|\leq k$. Since $D$ is dominating, every process~$p$ receives the input value of at least one process in~$D$, and can decide on such a value as an output. In total, at most $|D|\leq k$ values are decided.

\bigskip

Theorem~\ref{theo:impossibility-set-agreement}, together with properties~\ref{lemma-from-local-to-knowall} and~\ref{lemma-from-dynamic-to-knowall}, have implications for more traditional computational models:

\begin{corollary}\label{cor:consensus and k-set}
In the \local{} model, any algorithm that solves $k$-set agreement in a network~$G$ requires at least $r$~rounds, 
where $r$ is the smallest integer such that $\gamma(G^r)\leq k$. 
\end{corollary}

\begin{corollary}
In the dynamic network model, any  algorithm that solves $k$-set agreement against an adversary $A$ requires at least $r$~rounds, 
where $r$ is the smallest integer such that,  for every $\kagraph\in A$, $\gamma(\kagraphr)\leq k$.
\end{corollary}

\subsection{Approximate Agreement}
\label{subsec:appli2approxagreement}

Approximate agreement is the agreement task asking processes to output values that are as close as possible to each other, and, 
if all processes are given the same input value, 
then all processes should output that value. 
Specifically, let the input values be in $\{0,1\}$, let $\epsilon>0$, and let $k\geq1/\epsilon$ be an integer.
Then $\epsilon$-approximate agreement asks the $n$ processes to output values 
\[
v'_1,\dots,v'_n\in \left\{0,\frac{1}{k},\frac{2}{k},\ldots,\frac{k-1}{k},1\right\}
\]
such that $|v'_i-v'_j|\leq \epsilon$ for every $i,j$.  The associated input complex $\cal I$ is the same binary pseudosphere as for binary consensus, and the output complex can be defined by its facets: a facet for each set of output values satisfying $|v'_i-v'_j|\leq \epsilon$ for all $i,j$.
The carrier map $\Delta$ maps the all-$b$ input value facet to the all-$b$ output value facet, for every $b\in\{0,1\}$, and maps every other input facet to all the output facets.

Of course, if consensus can be solved, then $\epsilon$-approximate agreement can be solved with any $\epsilon$. The main point for studying approximate agreement is to determine the smallest $\epsilon$ for which $\epsilon$-approximate agreement is solvable, under the assumption that consensus is not solvable. In the proof of the next theorem, we show how topological arguments enable to resolve this problem easily in the \knowall\/ model, i.e., even when the communication between the processes is constrained.

\begin{theorem}\label{theo:approximateagreement}
Let $\kagraph$ be an instance of the \knowall{} model, and let $r\geq 0$. If consensus is not solvable  in~$r$ rounds in~$\kagraph$, then for every $\epsilon<\frac1{n-1}$, $\epsilon$-approximate agreement is also not solvable in~$r$ rounds in~$\kagraph$.
\end{theorem}

\begin{proof}
	Recall that, by Theorem~\ref{theo:impossibility-set-agreement}, consensus is solvable in $r$ rounds in $\kagraph$ if and only if $\kagraphr$ has a dominating node. 
	We show that no graph $\kagraphr$ without dominating node is able to solve $\epsilon$-approximate agreement for 
	$\epsilon<\frac1{n-1}$. 
	For $0\leq j\leq n$, let $\sigma_j$ be the facet of $\cI$
	defined by 
	\[
	\sigma_j=\{(p_i,1)\mid i\leq j\} \cup \{(p_i,0)\mid i> j\},
	\] 
	i.e., with $1$ as the input for the first $j$ processes, and $0$ for the rest.
	The facets $\sigma_j$ and $\sigma_{j+1}$ share an $(n-2)$-face. 
	The mapping $\Xi$ is not necessarily simplicial, so the images of $\sigma_j$ and $\sigma_{j+1}$ under $\Xi$ may not share such a face;
	nevertheless, we show that their images under $\Xi$ do intersect.
	Consider the sequence of facets $\sigma_0,\dots,\sigma_n$ in the input complex, 
	and its image $\Xi(\sigma_0),\dots,\Xi(\sigma_{n})$ in the protocol complex.
	Note that two consecutive simplexes, $\sigma_{j-1}$ and $\sigma_{j}$,
	differ only in the input of~$p_j$.
	Since $\kagraphr$ has no universal nodes, there is a process $p_k$ which does not receive messages from~$p_j$, 
	and thus there is a node $(p_k,v'_k)$ in the protocol complex
	that is shared by both $\Xi(\sigma_{j-1})$ and $\Xi(\sigma_{j})$.
	Applying the same argument for all values of $j=1,\dots,n$,
	we find that the image $(\Xi(\sigma_0),\ldots,\Xi(\sigma_{n}))$ in the protocol complex is path-connected, 
	in the sense that every two consecutive facets intersect.
	
	Assume that a protocol solves $\epsilon$-approximate agreement for some $\epsilon<\frac1{n-1}$. 
	It follows from Lemma~\ref{theo:iffcondition4solvability} that there is a simplicial map $\delta:\cP\rightarrow\cO$
	satisfying $\delta(\Xi(\sigma))\in\Delta(\sigma)$ for every $\sigma\in\cI$.
	Since $\delta$ is simplicial, 
	the image of $(\Xi(\sigma_0),\ldots,\Xi(\sigma_{n}))$ under $\delta$, 
	which is $\delta(\Xi(\sigma_0)),\ldots,\delta(\Xi(\sigma_{n}))$,
	is also path-connected.

	For every $\sigma_j$, let $V_j$ denote the possible output values for the processes in $\delta(\Xi(\sigma_j))$. 
	The definition of approximate agreement
	limits the domains $V_j$, and specifically, since
	\[\delta(\Xi(\sigma_0))=\{(p_i,0), i=1,\dots,n\},\]
	we have $V_0=\{0\}$,
	and similarly $V_{n}=\{1\}$.
	For all other values of $0<j<n$, $\epsilon$-approximate agreement imposes $V_j$ of range $\leq \epsilon$. Furthermore, $V_j\cap V_{j+1}\neq\emptyset$ for every $j\geq 0$ since every two consecutive facets $\delta(\Xi(\sigma_j))$ and $\delta(\Xi(\sigma_{j+1}))$
	share a vertex. 
	Achieving the desired contradiction is now simple.
	Each output value $v_i$ at a vertex of $\delta(\Xi(\sigma_0))$ must satisfy $v_i=0$.
	By connectivity, at least one such vertex is shared with $\delta(\Xi(\sigma_1))$,
	and thus every value $v'_i$ at a vertex  of $\delta(\Xi(\sigma_1))$ must satisfy $v'_i\leq\epsilon$.
	By induction, every value $v'_i$ at a vertex of $\delta(\Xi(\sigma_j))$ must satisfy $v'_i\leq j\epsilon$.
	Hence, every output value $v'_i$ at a vertex of $\delta(\Xi(\sigma_{n-1})$ must satisfy $v'_i\leq (n-1)\epsilon<1$.
	Once again, by connectivity there is a vertex $(p_i,v'_i)$ in $\delta(\Xi(\sigma_{n-1})\cap \delta(\Xi(\sigma_{n}))$.
	We have thus showed that the vertex $(p_i,v'_i)\in\delta(\Xi(\sigma_{n}))$ has output value $v'_i<1$,
	which is a contradiction for the specification of approximate agreement.
\end{proof}

Theorem~\ref{theo:approximateagreement}, together with properties~\ref{lemma-from-local-to-knowall} and~\ref{lemma-from-dynamic-to-knowall}, have implications for more traditional computational models. 

\begin{corollary}
In the \local{} model, any algorithm that solves $\epsilon$-approximate agreement in a network~$G$, with $\epsilon<\frac1{n-1}$, requires at least $r$~rounds, where $r$ is the smallest integer such that $G^r$ has a dominating vertex. 
\end{corollary}

\begin{corollary}
In the dynamic network model, any algorithm that solves $\epsilon$-approximate agreement against an adversary $A$, with $\epsilon<\frac1{n-1}$, requires at least $r$~rounds, 
where $r$ is the smallest integer such that for every $\kagraph\in A$, $\kagraphr$ has a dominating vertex. 
\end{corollary}

\paragraph{Remark.}
The bound in Theorem~\ref{theo:approximateagreement} is tight in the sense that there exists an instance $\kagraph$ of the \knowall{} model for which consensus is impossible, while $\frac1{n-1}$-approximate agreement is solvable. We show that there exists an information-flow graph $\kagraphr$ with no universal node, for which $\frac1{n-1}$-approximate agreement is solvable. 
$\kagraphr$ is simply obtained from a complete graph (each edge corresponds to two arcs oriented in opposite directions), from which all arcs in a single directed Hamiltonian cycle are removed. Indeed, this digraph has no universal node since every node has out-degree $n-2$. The $\frac1{n-1}$-approximate agreement algorithm is straightforward: each process chooses as  output value the average of all the input values it sees.
For correctness, let us consider an input assignment with $z$ input values~$0$, and $n-z$ input values~$1$.
Each process sees $n-1$ input values, out of which either $n-z-1$ or $n-z$ inputs are equal to~1, and thus the outputs are in the range 
	\[
	\left [\frac{n-z-1}{n-1},\frac{n-z}{n-1}\right],
	\]
	whose width is $\frac{1}{n-1}$.

\section{Conclusion and Further Work}

We have demonstrated that combinatorial topology is applicable to distributed network computing. Of course, this is just a first step, and further work will require to incorporate the specific features of different distributed network models, in order to capture the characteristics of each of them. For instance, fully capturing the popular \local{} model requires removing the structure awareness assumption, and studying the details of how the protocol complex evolves round after round. 

Incorporating asynchrony and failures into network computing, from a topological perspective, requires understanding the topological impact of simultaneously subdividing the facets, introducing holes resulting from $t$-resiliency, and introducing scissor cuts resulting from the presence of a network. This is definitely technically challenging, but our paper shows that there are no  conceptual obstacles preventing us from addressing these questions.

\paragraph{Acknowledgements: } The authors are thankful to Eli Gafni for fruitful discussions about solving $k$-set agreement in the \knowall\/ model, to Ran Gelles for discussions of carrier maps and scissor cuts, and to the reviewers of TCS journal for their help in improving the presentation of our results.


\let\OLDthebibliography\thebibliography
\renewcommand\thebibliography[1]{
	\OLDthebibliography{#1}
	\setlength{\parskip}{0pt}
	\setlength{\itemsep}{0pt plus 0.3ex}
}

\bibliographystyle{plain}
\bibliography{biblio}

\begin{thebibliography}{10}

\bibitem{ACPR19}
Manuel Alcantara, Armando Casta{\~{n}}eda, David Flores{-}Pe{\~{n}}aloza, and
  Sergio Rajsbaum.
\newblock The topology of look-compute-move robot wait-free algorithms with
  hard termination.
\newblock {\em Distributed Computing}, 32(3):235--255, 2019.

\bibitem{AlistarhAEGZ18}
Dan Alistarh, James Aspnes, Faith Ellen, Rati Gelashvili, and Leqi Zhu.
\newblock Why extension-based proofs fail.
\newblock In {\em Proceedings of the 51st Annual {ACM} {SIGACT} Symposium on
  Theory of Computing, {STOC}}, pages 986--996, 2019.

\bibitem{AlpernS85}
Bowen Alpern and Fred~B. Schneider.
\newblock Defining liveness.
\newblock {\em Inf. Process. Lett.}, 21(4):181--185, 1985.

\bibitem{AttiyaCHP19}
Hagit Attiya, Armando Casta{\~{n}}eda, Maurice Herlihy, and Ami Paz.
\newblock Bounds on the step and namespace complexity of renaming.
\newblock {\em {SIAM} J. Comput.}, 48(1):1--32, 2019.

\bibitem{BalliuBHORS19}
Alkida Balliu, Sebastian Brandt, Juho Hirvonen, Dennis Olivetti, Mika{\"{e}}l
  Rabie, and Jukka Suomela.
\newblock Lower bounds for maximal matchings and maximal independent sets.
\newblock In {\em 60th IEEE Symposium on Foundations of Computer Science
  (FOCS)}, 2019.

\bibitem{BarenboimEG18}
Leonid Barenboim, Michael Elkin, and Uri Goldenberg.
\newblock Locally-iterative distributed ($\delta+ 1$)-coloring below
  szegedy-vishwanathan barrier, and applications to self-stabilization and to
  restricted-bandwidth models.
\newblock In {\em Proceedings of the 2018 {ACM} Symposium on Principles of
  Distributed Computing, (PODC)}, pages 437--446, 2018.

\bibitem{BarenboimEPS12}
Leonid Barenboim, Michael Elkin, Seth Pettie, and Johannes Schneider.
\newblock The locality of distributed symmetry breaking.
\newblock In {\em 53rd IEEE Symposium on Foundations of Computer Science
  (FOCS)}, pages 321--330, 2012.

\bibitem{BhadraF12}
Sandeep Bhadra and Afonso Ferreira.
\newblock Computing multicast trees in dynamic networks and the complexity of
  connected components in evolving graphs.
\newblock {\em J. Internet Services and Applications}, 3(3):269--275, 2012.

\bibitem{BielyRSSW18}
Martin Biely, Peter Robinson, Ulrich Schmid, Manfred Schwarz, and Kyrill
  Winkler.
\newblock Gracefully degrading consensus and \emph{k}-set agreement in directed
  dynamic networks.
\newblock {\em Theor. Comput. Sci.}, 726:41--77, 2018.

\bibitem{BorowskyG93}
Elizabeth Borowsky and Eli Gafni.
\newblock Generalized {FLP} impossibility result for t-resilient asynchronous
  computations.
\newblock In {\em Proceedings of the Twenty-Fifth Annual {ACM} Symposium on
  Theory of Computing (STOC)}, pages 91--100, 1993.

\bibitem{BrandtFHKLRSU16}
Sebastian Brandt, Orr Fischer, Juho Hirvonen, Barbara Keller, Tuomo
  Lempi{\"{a}}inen, Joel Rybicki, Jukka Suomela, and Jara Uitto.
\newblock A lower bound for the distributed {L}ov{\'{a}}sz local lemma.
\newblock In {\em 48th ACM Symposium on Theory of Computing (STOC)}, pages
  479--488, 2016.

\bibitem{CastanedaFPRRT19talk}
Armando Casta{\~{n}}eda, Pierre Fraigniaud, Ami Paz, Sergio Rajsbaum, Matthieu
  Roy, and Corentin Travers.
\newblock A topological perspective on distributed network algorithms.
\newblock In {\em Structural Information and Communication Complexity - 26th
  International Colloquium, {SIROCCO}}, pages 3--18, 2019.

\bibitem{CastanedaR10}
Armando Casta{\~{n}}eda and Sergio Rajsbaum.
\newblock New combinatorial topology bounds for renaming: the lower bound.
\newblock {\em Distributed Computing}, 22(5-6):287--301, 2010.

\bibitem{CastanedaR12}
Armando Casta{\~{n}}eda and Sergio Rajsbaum.
\newblock New combinatorial topology bounds for renaming: The upper bound.
\newblock {\em J. {ACM}}, 59(1):3:1--3:49, 2012.

\bibitem{CasteigtsFGSY15}
Arnaud Casteigts, Paola Flocchini, Emmanuel Godard, Nicola Santoro, and
  Masafumi Yamashita.
\newblock On the expressivity of time-varying graphs.
\newblock {\em Theor. Comput. Sci.}, 590:27--37, 2015.

\bibitem{CasteigtsFQS12}
Arnaud Casteigts, Paola Flocchini, Walter Quattrociocchi, and Nicola Santoro.
\newblock Time-varying graphs and dynamic networks.
\newblock {\em International Journal of Parallel, Emergent and Distributed
  Systems}, 27(5):387--408, 2012.

\bibitem{ChangLP18}
Yi{-}Jun Chang, Wenzheng Li, and Seth Pettie.
\newblock An optimal distributed $({\Delta}+1)$-coloring algorithm?
\newblock In {\em 50th ACM Symposium on Theory of Computing (STOC)}, pages
  445--456, 2018.

\bibitem{Charron-BostFN15}
Bernadette Charron{-}Bost, Matthias F{\"{u}}gger, and Thomas Nowak.
\newblock Approximate consensus in highly dynamic networks: The role of
  averaging algorithms.
\newblock In {\em Automata, Languages, and Programming - 42nd International
  Colloquium, (ICALP)}, pages 528--539, 2015.

\bibitem{Charron-BostFN16}
Bernadette Charron{-}Bost, Matthias F{\"{u}}gger, and Thomas Nowak.
\newblock Fast, robust, quantizable approximate consensus.
\newblock In {\em 43rd International Colloquium on Automata, Languages, and
  Programming, {ICALP}}, pages 137:1--137:14, 2016.

\bibitem{Charron-BostS09}
Bernadette Charron{-}Bost and Andr{\'{e}} Schiper.
\newblock The heard-of model: computing in distributed systems with benign
  faults.
\newblock {\em Distributed Computing}, 22(1):49--71, 2009.

\bibitem{ChaudhuriHLT00}
Soma Chaudhuri, Maurice Herlihy, Nancy~A. Lynch, and Mark~R. Tuttle.
\newblock Tight bounds for \emph{k}-set agreement.
\newblock {\em J. {ACM}}, 47(5):912--943, 2000.

\bibitem{CouloumaGP15}
Etienne Coulouma, Emmanuel Godard, and Joseph~G. Peters.
\newblock A characterization of oblivious message adversaries for which
  consensus is solvable.
\newblock {\em Theor. Comput. Sci.}, 584:80--90, 2015.

\bibitem{FischerGK17}
Manuela Fischer, Mohsen Ghaffari, and Fabian Kuhn.
\newblock Deterministic distributed edge-coloring via hypergraph maximal
  matching.
\newblock In {\em 58th {IEEE} Annual Symposium on Foundations of Computer
  Science (FOCS)}, pages 180--191, 2017.

\bibitem{FischerLP85}
Michael~J. Fischer, Nancy~A. Lynch, and Mike Paterson.
\newblock Impossibility of distributed consensus with one faulty process.
\newblock {\em J. {ACM}}, 32(2):374--382, 1985.

\bibitem{FuggerNS18}
Matthias F{\"{u}}gger, Thomas Nowak, and Manfred Schwarz.
\newblock Tight bounds for asymptotic and approximate consensus.
\newblock In {\em Proceedings of the 2018 {ACM} Symposium on Principles of
  Distributed Computing, {PODC}}, pages 325--334, 2018.

\bibitem{Ghaffari16}
Mohsen Ghaffari.
\newblock An improved distributed algorithm for maximal independent set.
\newblock In {\em 27th {ACM-SIAM} Symposium on Discrete Algorithms (SODA)},
  pages 270--277, 2016.

\bibitem{GhaffariKM17}
Mohsen Ghaffari, Fabian Kuhn, and Yannic Maus.
\newblock On the complexity of local distributed graph problems.
\newblock In {\em 49th ACM Symposium on Theory of Computing (STOC)}, pages
  784--797, 2017.

\bibitem{GodardP16}
Emmanuel Godard and Eloi Perdereau.
\newblock k-set agreement in communication networks with omission faults.
\newblock In {\em 20th International Conference on Principles of Distributed
  Systems (OPODIS)}, pages 8:1--8:17, 2016.

\bibitem{GoosHS17}
Mika G{\"{o}}{\"{o}}s, Juho Hirvonen, and Jukka Suomela.
\newblock Linear-in-{$\Delta$} lower bounds in the {LOCAL} model.
\newblock {\em Distributed Computing}, 30(5):325--338, 2017.

\bibitem{HarrisSS16}
David~G. Harris, Johannes Schneider, and Hsin{-}Hao Su.
\newblock Distributed $({\Delta}+1)$-coloring in sublogarithmic rounds.
\newblock In {\em 48th ACM Symposium on Theory of Computing (STOC)}, pages
  465--478, 2016.

\bibitem{HerlihyKR13}
Maurice Herlihy, Dmitry Kozlov, and Sergio Rajsbaum.
\newblock {\em Distributed Computing Through Combinatorial Topology}.
\newblock Morgan Kaufmann, 2013.

\bibitem{HerlihyR94}
Maurice Herlihy and Sergio Rajsbaum.
\newblock Set consensus using arbitrary objects.
\newblock In {\em Proceedings of the Thirteenth Annual {ACM} Symposium on
  Principles of Distributed Computing (PODC)}, pages 324--333, 1994.

\bibitem{HerlihyR00}
Maurice Herlihy and Sergio Rajsbaum.
\newblock Algebraic spans.
\newblock {\em Mathematical Structures in Computer Science}, 10(4):549--573,
  2000.

\bibitem{HerlihyRT09}
Maurice Herlihy, Sergio Rajsbaum, and Mark~R. Tuttle.
\newblock An axiomatic approach to computing the connectivity of synchronous
  and asynchronous systems.
\newblock {\em Electr. Notes Theor. Comput. Sci.}, 230:79--102, 2009.

\bibitem{HerlihyS93}
Maurice Herlihy and Nir Shavit.
\newblock The asynchronous computability theorem for t-resilient tasks.
\newblock In {\em 25th {ACM} Symposium on Theory of Computing (STOC)}, pages
  111--120, 1993.

\bibitem{HerlihyS99}
Maurice Herlihy and Nir Shavit.
\newblock The topological structure of asynchronous computability.
\newblock {\em J. {ACM}}, 46(6):858--923, 1999.

\bibitem{KuhnLO10}
Fabian Kuhn, Nancy~A. Lynch, and Rotem Oshman.
\newblock Distributed computation in dynamic networks.
\newblock In {\em 42nd {ACM} Symposium on Theory of Computing (STOC)}, pages
  513--522, 2010.

\bibitem{KuhnMW16}
Fabian Kuhn, Thomas Moscibroda, and Roger Wattenhofer.
\newblock Local computation: Lower and upper bounds.
\newblock {\em J. {ACM}}, 63(2):17:1--17:44, 2016.

\bibitem{KuhnOM11}
Fabian Kuhn, Yoram Moses, and Rotem Oshman.
\newblock Coordinated consensus in dynamic networks.
\newblock In {\em 30th {ACM} Symposium on Principles of Distributed Computing
  (PODC)}, pages 1--10, 2011.

\bibitem{KuhnO2011}
Fabian Kuhn and Rotem Oshman.
\newblock Dynamic networks: Models and algorithms.
\newblock {\em SIGACT News}, 42(1):82--96, 2011.

\bibitem{Linial92}
Nathan Linial.
\newblock Locality in distributed graph algorithms.
\newblock {\em {SIAM} J. Comput.}, 21(1):193--201, 1992.

\bibitem{MendesTH14}
Hammurabi Mendes, Christine Tasson, and Maurice Herlihy.
\newblock Distributed computability in {B}yzantine asynchronous systems.
\newblock In {\em 46th Symposium on Theory of Computing (STOC)}, pages
  704--713, 2014.

\bibitem{Nowak0W19}
Thomas Nowak, Ulrich Schmid, and Kyrill Winkler.
\newblock Topological characterization of consensus under general message
  adversaries.
\newblock In {\em Proceedings of the 2019 {ACM} Symposium on Principles of
  Distributed Computing, (PODC)}, pages 218--227, 2019.

\bibitem{Peleg2000}
David Peleg.
\newblock {\em Distributed Computing: A Locality-Sensitive Approach}.
\newblock SIAM, Philadelphia, PA, 2000.

\bibitem{RajsbaumRT08}
Sergio Rajsbaum, Michel Raynal, and Corentin Travers.
\newblock The iterated restricted immediate snapshot model.
\newblock In {\em 14th Int. Conference on Computing and Combinatorics
  (COCOON)}, pages 487--497, 2008.

\bibitem{SakavalasTseng2019}
Dimitris Sakavalas and Lewis Tseng.
\newblock Network topology and fault-tolerant consensus.
\newblock Synthesis Lectures on Distributed Computing Theory, 2019.

\bibitem{SaksZ93}
Michael~E. Saks and Fotios Zaharoglou.
\newblock Wait-free k-set agreement is impossible: the topology of public
  knowledge.
\newblock In {\em 25th {ACM} Symposium on Theory of Computing (STOC)}, pages
  101--110, 1993.

\bibitem{Suomela13}
Jukka Suomela.
\newblock Survey of local algorithms.
\newblock {\em {ACM} Comput. Surv.}, 45(2):24:1--24:40, 2013.

\end{thebibliography}

\appendix
\bigskip
\centerline{\Large\bf A P P E N D I X}
\section{Basic Topological Concepts}
\label{appendix:topological basics}

A  \emph{simplicial complex} is a finite set $V$ along with a collection $\cK$  of nonempty subsets of $V$ closed under containment (i.e., if $A\in \cK$ and $\emptyset\neq B\subset A$, then $B\in \cK$).
An element of $V$ is called a \emph{vertex} of $\cK$, and the vertex set of $\cK$ is denoted by $V(\cK)=V$.
Each set in $\cK$ is called a \emph{simplex}.
A subset of a simplex is called a \emph{face} of that simplex.
The \emph{dimension} $\dim \sigma$ of a simplex $\sigma$ is
one less than the number of elements of $\sigma$, i.e., $|\sigma|-1$.
We use ``$d$-face'' as shorthand for ``$d$-dimensional face''.
A simplex $\sigma$ in $\cK$ is called a \emph{facet} of $\cK$ if $\sigma$ is not contained in any other simplex.
Note that a set of facets uniquely defines a simplicial complex.
The dimension of a complex is the largest dimension of any of its facets.
A complex is \emph{pure} if all its facets have the same dimension.
For two complexes $\cK$ and $\cL$, if $\cK \subseteq \cL$, we say $\cK$ is a \emph{subcomplex} of $\cL$.
When clear from the context, we refer to the union of one or more simplexes as a complex. Such a complex should be understood as encompassing those simplexes together with all of their faces. In particular, 
we sometimes use a simplex $\sigma$ as shorthand for the complex defined by its power set, i.e., the complex formed by $\sigma$ and all its faces.

The \emph{$d$-skeleton} of a complex $\cK$, denoted $\skel^{d}\cK$, is the subcomplex of $\cK$ composed of all the faces of $\cK$ of dimension at most $d$. If the complex $\cK$ is composed of a single $d$-simplex and all its faces, the \emph{boundary complex} of $\cK$ is its $(d-1)$-skeleton.

Let $\cK$ and $\cL$ be complexes.
A \emph{vertex map} is a function $h:V(\cK)\to V(\cL)$.
If $h$ also carries simplexes of $\cK$ to simplexes of $\cL$,
it is called a \emph{simplicial map}.
We add one or more \emph{labels} to the vertices,
where a labeling is a mapping $\lambda: V \rightarrow D$, for some arbitrary domain $D$.
In this paper, we always have the labeling $\name(\cdot)$, which associates each vertex with a unique name in $\{1,\ldots,n\}$.
In all the complexes we consider, each simplex is properly colored by these names:
if $u$ and $v$ are distinct vertices of a simplex $\sigma$,
then $\name(u) \neq \name(v)$.
In this case we say that the complexes are \emph{chromatic} (with respect to the $\name$ labeling).
A simplicial map $h$ is \emph{chromatic} if it preserves names, i.e., 
$\name(h(v)) = \name(v)$ for any vertex~$v$.
In this paper,
all simplicial maps between chromatic complexes are chromatic.
By extension, given any simplex $\sigma$, we denote by $\name(\sigma)$ the set of names of the vertices appearing in $\sigma$. 

Given two complexes $\cK$ and $\cL$, a \emph{carrier map} $\Phi$ maps each simplex $\sigma \in \cK$ to a subcomplex $\Phi(\sigma)$ of $\cL$, such that for every two simplexes $\tau$ and $\tau'$ in $\cK$ that satisfy $\tau \subseteq \tau'$, we have $\Phi(\tau) \subseteq \Phi(\tau')$. 

Roughly speaking, a \emph{geometric realization} $|\cK|$ of a simplicial complex $\cK$ is a geometric object defined as follows.
Each vertex in $V(\cK)$ is mapped to a point in a Euclidean space, such that the images of the vertices are affinely independent. Each simplex is represented by a polyhedron, which is the convex hull of points representing its vertices.
Figure~\ref{fig:scisor} displays the geometric representations of several simplicial complexes.

Let $k$ be a positive integer. We say that a complex has \emph{a hole in dimension~$k$} if the $k$-sphere~$S^k$ embedded in a geometric realization of the complex cannot be continuously contracted to  a single point within that realization. Informally, a complex is \emph{$k$-connected} if it has no holes in dimension~$k$. More formally, a complex~$\cK$ is \emph{$k$-connected} if every continuous map $h: S^k \to |\cK|$ can be extended to a continuous map $h': D^{k+1} \to |\cK|$ where $D^{k+1}$ denotes the~$(k+1)$-disk.
In dimension $0$, this property simply states that any two points can be linked by a path, i.e., the complex is \emph{path-connected}. In dimension~$1$, it states that any loop can be filled into a disk, i.e., the complex is simply connected.
By convention, a $(-1)$-connected complex is just a non-empty complex, and every complex is $d$-connected for every~$d < -1$. 

Finally, given a set $I$, a \emph{pseudosphere} $\Psi(\{1,\ldots,n\},I)$ is the $(n-1)$-dimensional complex defined as follows:
(1)~every pair $(i,v)$ with $i\in \{1,\ldots,n\}$, and $v\in I$ is a vertex, and 
(2)~every $\{(1, v_1), \hdots, (n, v_n) \}$ with $v_i \in I$ for every $i\in[n]$ is a facet.
Pseudospheres offer a convenient way to describe all possible initial configurations when each process input is an arbitrary value from~$I$.

\end{document}